\documentclass[11pt,a4paper]{article}

\usepackage{amssymb} 
\usepackage{amsmath}     
\usepackage{amsthm}
\usepackage{todonotes}
\usepackage{mathtools}
\usepackage{braket}
\usepackage{subfigure}

\usepackage[ruled,vlined,linesnumbered]{algorithm2e}
\SetKw{KwNot}{not}
\SetKw{KwOR}{or}
\SetKw{KwAND}{and}
\SetKw{KwExitFor}{exit for-loop}

\usepackage{todonotes}

\usepackage{authblk}

\usepackage{a4wide}
\usepackage{graphicx}
\usepackage{hyperref}
\newcommand{\Rset}{\mathbb{R}}

\newcommand{\X}{{\mathcal{X}}}

\newtheorem{thm}{Theorem}
\newtheorem{lem}{Lemma}

\newtheorem{cor}{Corollary}
\newtheorem{prop}{Proposition}

\newtheorem{obs}{Observation}

\begin{document}

\title{Combinatorial two-stage minmax regret problems under interval uncertainty}

\author[1]{Marc Goerigk}

\author[2]{Adam Kasperski}
\author[2]{Pawe{\l} Zieli\'nski}

\affil[1]{Network and Data Science Management, University of Siegen, Germany, \texttt{marc.goerigk@uni-siegen.de}}
\affil[2]{
Wroc{\l}aw  University of Science and Technology, Wroc{\l}aw, Poland\\
            \texttt{\{adam.kasperski,pawel.zielinski\}@pwr.edu.pl}}

 \date{}

\maketitle

\begin{abstract}

In this paper a class of combinatorial optimization problems is discussed. It is assumed that a feasible solution can be constructed in two stages. In the first stage the objective function costs are known while in the second stage they are uncertain and belong to an interval uncertainty set. In order to choose a solution, the minmax regret criterion is used. Some general properties of the problem are established and results for two particular problems, namely the shortest path and the selection problem, are shown.

\end{abstract}


\section{Introduction}

Consider the following deterministic single-stage \emph{combinatorial optimization problem}:
\begin{equation}
{\rm opt}^{\mathcal{P}}(\pmb{c})=\min_{\pmb{x}\in \mathcal{X}}\;  \pmb{c}^T\pmb{x}, \tag{$\mathcal{P}$}
\end{equation}
where $\mathcal{X}\subseteq \{0,1\}^n$ is a set of \emph{feasible solutions} and $\pmb{c}\in \Rset_+^n$ is a vector of nonnegative objective function costs. Typically, $\mathcal{X}$ is described by a system of linear constraints involving binary variables $x_i$, $i\in [n]$ (we will use the notation $[n]=\{1,\dots,n\}$), which leads to a 0-1 programming problem.  Solution $\pmb{x}\in \mathcal{X}$ can be interpreted as a characteristic vector of some finite element set $E$. For example, $E$ can be a set of edges of a graph $G=(V,E)$ and $\pmb{x}$ describes some objects in $G$ such as paths, trees, matchings etc.

In many practical applications the vector of objective function costs $\pmb{c}$ is uncertain and it is only known to belong to an uncertainty set $\mathcal{U}$. There are various methods of defining $\mathcal{U}$ which depend on the application and the information available. Among the easiest and most common is the \emph{interval uncertainty representation}
(see, e.g.,~\cite{KY97}), in which $c_i\in [\underline{c}_i, \overline{c}_i]$ for each $i\in [n]$ and $\mathcal{U}=\prod_{i\in [n]} [\underline{c}_i,\overline{c}_i]\subseteq\Rset_+^n$. In order to choose a solution, for a specified $\mathcal{U}$, one can apply a robust decision criterion, which takes into account the worst cost realizations. Under the interval uncertainty representation, the \emph{minmax regret} criterion (also called \emph{Savage criterion}~\cite{SA51}) has attracted a considerable attention in the literature.  The \emph{regret} of a given solution $\pmb{x}\in \mathcal{X}$ under a cost scenario $\pmb{c}\in \mathcal{U}$ is the quantity $\pmb{c}^T\pmb{x}-{\rm opt}^{\mathcal{P}}(\pmb{c})$. It expresses a deviation of solution $\pmb{x}$ from the optimum and can be interpreted as the maximal opportunity loss after $\pmb{x}$ is implemented. In the \emph{single-stage minmax regret} version of $\mathcal{P}$ we seek a solution minimizing the maximum regret, i.e. we study the following problem:
\begin{equation}
\label{regrP}
\min_{\pmb{x}\in \mathcal{X}}\max_{\pmb{c}\in \mathcal{U}} (\pmb{c}^T\pmb{x}-{\rm opt}^{\mathcal{P}}(\pmb{c})) \tag{\textsc{SStR}~$\mathcal{P}$}.
\end{equation}

The \textsc{SStR}~$\mathcal{P}$ problem has been discussed in a number of papers, for example when $\mathcal{P}$ is the minimum spanning tree~\cite{YKP01}, the shortest path~\cite{KPY01}, the minimum $s$-$t$ cut~\cite{ABV08}, the minimum assignment~\cite{ABV05, PA11}, or the selection~\cite{AV01, C04} problem (in this case $\mathcal{X}=\{\pmb{x}\in \{0,1\}^n: x_1+\dots+x_n=p\}$ for some fixed $p\in [n]$). Surveys of known results in this area can be found in~\cite{ABV09, KZ16b}. Unfortunately, \textsc{SStR}~$\mathcal{P}$ turned out to be NP-hard for all previously mentioned problems~\cite{AH04, AL04, Z04, KZ06a}, with a notable exception when $\mathcal{P}$ is the selection problem, for which polynomial algorithms were established in~\cite{AV01, C04}. The \textsc{SStR}~$\mathcal{P}$ problem has some well known general properties. There is a nice characterization of scenario $\pmb{c}\in \mathcal{U}$ maximizing the regret of a given solution $\pmb{x}$ (called a \emph{worst-case scenario} for $\pmb{x}$), namely, $c_i=\overline{c}_i$ if $x_i=1$ and $c_i=\underline{c}_i$ if $x_i=0$ for each $i\in [n]$. Notice that this scenario depends only on $\pmb{x}$ and the problem of computing the maximum regret of a given solution has the same complexity as $\mathcal{P}$. Also, there is a general 2-approximation algorithm known for \textsc{SStR}~$\mathcal{P}$, under the assumption that $\mathcal{P}$ is polynomially solvable~\cite{KZ06,CO10,chassein2015new}. We get a 2-approximate solution by solving $\mathcal{P}$ under the so-called \emph{midpoint scenario} $\pmb{c}^m\in \mathcal{U}$ such that $c_i^m=(\underline{c}_i+\overline{c}_i)/2$ for each $i\in [n]$. Exact algorithms for solving \textsc{SStR}~$\mathcal{P}$ are based on compact mixed-integer programming (MIP) formulations (see, e.g.,~\cite{YKP01}), when $\mathcal{P}$ has a special structure, or constraint generation technique in general (see, e.g.,~\cite{PA13}).

In some applications a solution from $\mathcal{X}$ can be constructed in two stages. Namely, a partial solution is chosen now (in the first stage) and is completed in the future (in the second stage). The current, first-stage costs are known while the future second-stage costs are uncertain and belong to an uncertainty set $\mathcal{U}$. However, the partial solution can be completed after a second-stage cost scenario is revealed. The problem consists in computing a best  first-stage solution, which corresponds to the decision which must be made now.
The two-stage approach has a long tradition in stochastic optimization (see, e.g.,~\cite{KM05}). When a probability distribution in $\mathcal{U}$ is unknown, then a robust two-stage version of $\mathcal{P}$ can be considered. First such a model was discussed in~\cite{KMU08} for the assignment problem. This approach was also applied to the minimum spanning tree~\cite{KZ11} and the selection problems~\cite{CGKZ18}. In these papers the robust minmax criterion has been applied, i.e. a first-stage solution is determined minimizing the largest total first and second-stage cost.

In this paper we wish to investigate the two-stage version of problem $\mathcal{P}$ under the interval uncertainty representation. Namely, for each second-stage cost an interval of its possible values is provided.  We use the minmax regret criterion to choose a solution. The interpretation of this problem is the same as in the case of \textsc{SStR}~$\mathcal{P}$. We seek a first-stage solution, which minimizes the maximum regret, i.e. the maximum distance to a best first-stage solution.
 We will show that this problem has different properties than its single-stage counterpart. In particular, there is no easy characterization of a worst-case scenario of a given first-stage solution, although  there is still a worst-case scenario which is extreme (the second-stage costs take their upper or lower bounds under this scenario). In fact, the problem of computing the maximum regret can be NP-hard, even if $\mathcal{P}$ is solvable in polynomial time.
 Also, the midpoint heuristic does not guarantee any approximation ratio in general. We will show a general method of solving the problem, which is based on a MIP formulation. We then study two special cases, when $\mathcal{P}$ is the shortest path and the selection problem.
 
 This paper is organized as follows. In Section~\ref{secpf} we state the problem. We also consider three inner problems, in particular the problem of computing the maximum regret of a given first-stage solution. In Section~\ref{secmip}, we construct MIP formulations, which can be used to compute exact or approximate solutions. Section~\ref{secsp} is devoted to the two-stage version of the shortest path problem. We proceed with the study of two variants of this problem, which have different computational properties. We show that both computing an optimal first-stage solution and the maximum regret of a given first-stage solution are NP-hard. In Section~\ref{secsel} we discuss the selection problem.  We show that for this problem the maximum regret of a given first-stage solution can be computed in polynomial time and the optimal first-stage solution can by determined by using a compact MIP formulation. We also propose a greedy heuristic for this problem. Finally, the paper is concluded and further research questions are pointed out in Section~\ref{sec:conclusions}.

\section{Problem formulation}
\label{secpf}

In this paper we assume that a solution from $\mathcal{X}$ can be built in two-stages. Given a vector $\pmb{x}\in\{0,1\}^n$, let 
\begin{equation}
\label{ract}
\mathcal{R}(\pmb{x})=\{\pmb{y}\in \{0,1\}^n: \pmb{x}+\pmb{y}\in \mathcal{X}\}
\end{equation}
be the set of \emph{recourse actions} for $\pmb{x}$. Vector $\pmb{y}\in \mathcal{R}(\pmb{x})$ is a completion of the partial solution $\pmb{x}$ to a feasible one.
Let
$$\mathcal{X}'=\{\pmb{x}\in\{0,1\}^n:\mathcal{R}(\pmb{x})\neq \emptyset\}$$ 
be the set of \emph{feasible first-stage solutions}. 
 Define $\mathcal{Z}=\{(\pmb{x},\pmb{y})\in\{0,1\}^{2n}: \pmb{x}\in \mathcal{X}', \pmb{y}\in \mathcal{R}(\pmb{x})\}$ as the set of all possible combinations between partial first-stage solutions and recourse actions. Given a first-stage cost vector $\pmb{C}\in \Rset_+^n$ and a second-stage cost vector $\pmb{c}\in \Rset_+^n$, we consider the following \emph{two-stage problem}:
\begin{equation}
{\rm Opt}(\pmb{c})=\min_{(\pmb{x},\pmb{y})\in \mathcal{Z}}  (\pmb{C}^T\pmb{x}+\pmb{c}^T\pmb{y}). \tag{\textsc{TSt}~
$\mathcal{P}$}
\end{equation}
Given $\pmb{x}\in \mathcal{X}'$ and $\pmb{c}\in \Rset_+^n$, we will also examine the following \emph{incremental problem}:
\begin{equation}
{\rm Inc}(\pmb{x},\pmb{c})=\pmb{C}^T\pmb{x}+\min_{\pmb{y}\in \mathcal{R}(\pmb{x})} \pmb{c}^T\pmb{y} \tag{\textsc{Inc}~$\mathcal{P}$}
\end{equation}
in which we seek a best recourse action for $\pmb{x}\in \mathcal{X}'$ and $\pmb{c}$. The quantity ${\rm Inc}(\pmb{x},\pmb{c})-{\rm Opt}(\pmb{c})$ is called the \emph{regret} of $\pmb{x}$ under $\pmb{c}$.
  Suppose that the second-stage costs are uncertain and we only know that $c_i\in [\underline{c}_i,\overline{c}_i]$ for each $i\in [n]$. We thus consider the \emph{interval uncertainty representation} $\mathcal{U}=\prod_{i\in [n]} [\underline{c}_i,\overline{c}_i]$. Each possible second-stage cost vector $\pmb{c}\in \mathcal{U}$ is called a \emph{scenario}. Let us define the maximum regret of a given first-stage solution $\pmb{x}\in \mathcal{X}'$ as follows:
 \begin{equation}
 \label{defz}
  	Z(\pmb{x})=\max_{\pmb{c}\in \mathcal{U}}\left({\rm Inc}(\pmb{x},\pmb{c})-{\rm Opt}(\pmb{c})\right).
\end{equation}
 A scenario which maximizes the right hand side of~(\ref{defz}) is  called a \emph{worst-case scenario} for $\pmb{x}$.
 In this paper we study the following \emph{two-stage minmax regret problem}:
 \begin{equation}
 \min_{\pmb{x}\in \mathcal{X}'} Z(\pmb{x})=\min_{\pmb{x}\in \mathcal{X}'}\max_{\pmb{c}\in \mathcal{U}}\min_{\pmb{y}\in \mathcal{R}(\pmb{x})}\left(\pmb{C}^T\pmb{x}+\pmb{c}^T\pmb{y}-{\rm Opt}(\pmb{c})\right). \tag{\textsc{TStR}~$\mathcal{P}$}
 \end{equation}

Let us illustrate \textsc{TStR}~$\mathcal{P}$ when $\mathcal{P}$ is the \textsc{Shortest Path} problem shown in 
Figure~\ref{figex1}~a. In this case, $\mathcal{X}$ is the set of characteristic vectors of the simple $s-t$ paths in a given network $G=(V,A)$.

\begin{figure}[ht]
	\centering
	\includegraphics[height=4cm]{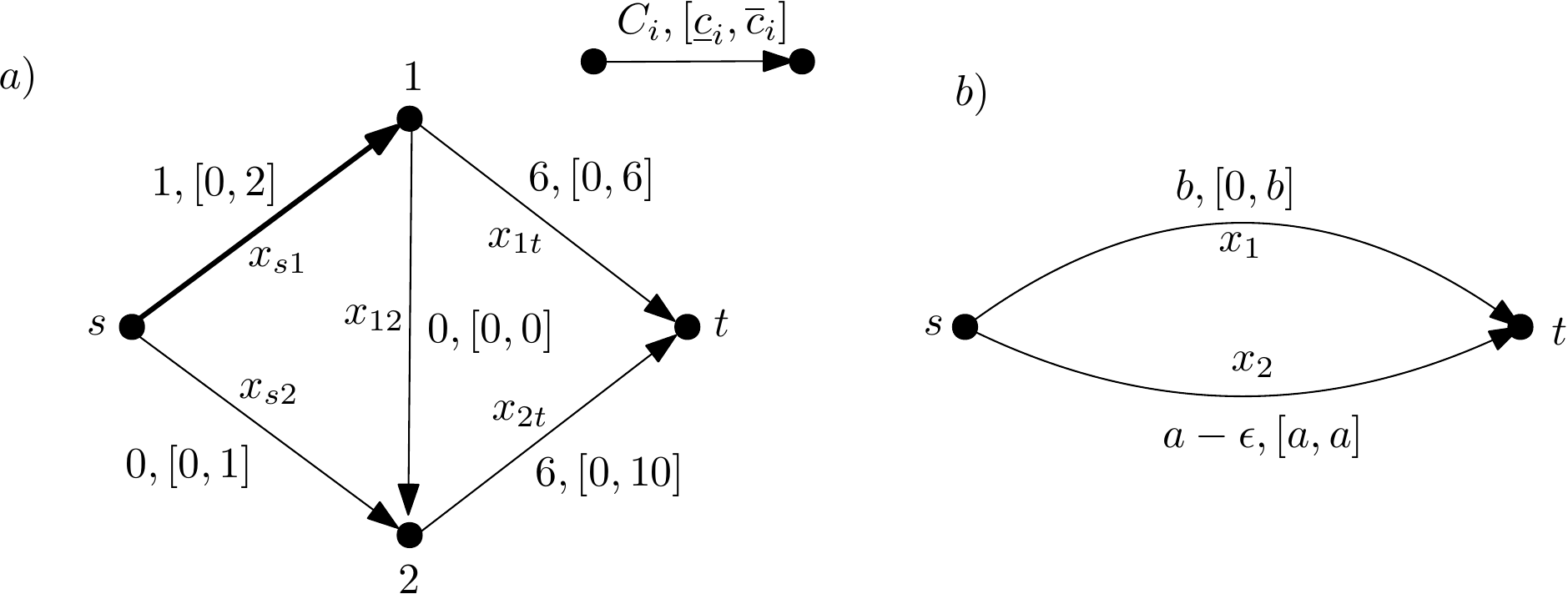}
	\caption{Two instances of \textsc{TStR Shortest Path}. In the instance b) we assume that $b\gg 2a \gg \varepsilon>0$.}\label{figex1}
\end{figure}

Let $\pmb{x}=(x_{s1},x_{s2}, x_{12}, x_{1t}, x_{2t})^T\in \mathcal{X}'$ denote a first-stage solution to the instance~a). A candidate solution is $\pmb{x}'=(1,0,0,0,0)^T$, in which the arc $(s,1)$ is selected in the first stage. Under scenario $\underline{\pmb{c}}=(0,0,0,0,0)^T$, this partial solution can be completed to a path by choosing arc $(1,t)$ with total costs $1+0=1$. As $\textrm{Opt}(\underline{\pmb{c}}) = 0$, the maximum regret of $\pmb{x}'$ is at least~1. It can be verified that there is no other scenario that results in a higher regret.
As a second example, the first-stage solution $(0,1,0,0,0)^T$ has the maximum regret equal to~10. Indeed in a worst scenario $\pmb{c}'$ the cost of $(2,t)$ is set to~10 and the costs of the remaining arcs are set to~0. The arcs $(2,t)$ must be selected in the second stage and ${\rm opt}(\pmb{c}')=0$, which results in the regret of~10.
Notice also that solution $(0,0,0,0,0)^T$, when no arc is selected in the first stage, has the maximum regret equal to~2, achieved by using scenario $\overline{\pmb{c}}$. A full enumeration reveals that $\pmb{x}'$ is in fact optimal.

The instance in Figure~\ref{figex1}~b demonstrates that the mid-point heuristic does not guarantee the approximation ratio of~2 for \textsc{TStR}~$\mathcal{P}$. Indeed, if we solve the \textsc{TSt Shortest Path} problem for the second-stage midpoint scenario $\pmb{c}^m=(b/2,a)^T$, then we get solution $\pmb{x}^m=(0,1)^T$ with $Z(\pmb{x}^m)=a-\epsilon$. But the optimal first-stage solution is $\pmb{x}=(0,0)^T$ with $Z(\pmb{x})=\epsilon$. Hence the ratio $Z(\pmb{x}^m)/Z(\pmb{x})=(a-\epsilon)/\epsilon$ can be arbitrarily large.

\section{Mixed integer programming formulations}
\label{secmip}

In this section we construct mixed integer programming formulations for computing the maximum regret $Z(\pmb{x})$ of a given first-stage solution $\pmb{x}\in \mathcal{X}'$ and solving the \textsc{TStR}~$\mathcal{P}$ problem. In particular, we will show that for each $\pmb{x}\in \mathcal{X}'$, there exists a worst-case scenario which is extreme, i.e. which belongs to
 $\prod_{i\in [n]}\{\underline{c}_i, \overline{c}_i\}$. In the following, we will use $\underline{\pmb{c}}$ to denote the scenario $(\underline{c}_1,\dots,\underline{c}_n)^T$.
Fix $(\pmb{u},\pmb{v})\in \mathcal{Z}$ and define
$$Z_{(\pmb{u},\pmb{v})}(\pmb{x})=\pmb{C}^T\pmb{x}-\pmb{C}^T\pmb{u}+\max_{\pmb{c}\in \mathcal{U}} \min_{\pmb{y}\in \mathcal{R}(\pmb{x})} \pmb{c}^T(\pmb{y}-\pmb{v}).$$
It is easy to verify that
\begin{equation}
\label{zuv}
Z(\pmb{x})=\max_{(\pmb{u},\pmb{v})\in \mathcal{Z}} Z_{(\pmb{u},\pmb{v})}(\pmb{x}).
\end{equation}

\begin{prop}
\label{propwc1}
	It holds that
	\begin{equation}\label{eqprop1}
	Z_{(\pmb{u},\pmb{v})}(\pmb{x})=\pmb{C}^T\pmb{x}-\pmb{C}^T\pmb{u}-\underline{\pmb{c}}^T\pmb{v}+ \min_{\pmb{y}\in \mathcal{R}(\pmb{x})} \pmb{c}_{\pmb{v}}^{T}\pmb{y},
	\end{equation}
	where scenario $\pmb{c}_{\pmb{v}}\in \mathcal{U}$ is such that  $c_{\pmb{v}i}=\underline{c}_i$ if $v_i=1$ and $c_{\pmb{v}i}=\overline{c}_i$ if $v_i=0$.
\end{prop}
\begin{proof}
	Write $Z_{(\pmb{u},\pmb{v})}(\pmb{x})=\pmb{C}^T\pmb{x}-\pmb{C}^T\pmb{u}+t^*$, where $t^*$ is computed by solving the following problem:
\begin{align}
\max \  & t \label{c03}\\ 
\text{s.t. } & t \le \pmb{c}^T(\pmb{y}-\pmb{v}) & \forall \pmb{y}\in \mathcal{R}(\pmb{x})\\
& c_i\in [\underline{c}_i,\overline{c}_i] & \forall i\in [n] \label{c04}
\end{align}
An optimal solution to~(\ref{c03})-(\ref{c04}) can be determined as follows: for each $i\in [n]$, if $v_i=1$, then $c_i=\underline{c}_i$ (because $y_i-v_i\leq 0$) and if $v_i=0$, then $c_i=\overline{c}_i$ (because $y_i-v_i\geq 0)$. This yields the scenario $\pmb{c}_{\pmb{v}}$. Observe also that $\pmb{c}_{\pmb{v}}^T(\pmb{y}-\pmb{v})=-\underline{\pmb{c}}^T\pmb{v}+\min_{\pmb{y}\in \mathcal{R}(\pmb{x})} \pmb{c}_{\pmb{v}}^{T}\pmb{y}$, which completes the proof.
\end{proof}
Proposition~\ref{propwc1} implies the following corollary:
\begin{cor} 
\label{wcsc}
For each $\pmb{x}\in \mathcal{X}'$, there is a worst-case extreme scenario $\pmb{c}\in \prod_{i\in [n]} \{\underline{c}_i,\overline{c}_i\}$.
\end{cor}

The result stated in Corollary~\ref{wcsc} is analogous to the known result for the single-stage \textsc{SStR}~$\mathcal{P}$ problem (see, e.g.,~\cite{ABV09}).
However, in the two-stage model the worst-case scenario for $\pmb{x}\in \mathcal{X}'$ is not completely characterized by $\pmb{x}$. In order to compute $Z(\pmb{x})$ one needs to find  $(\pmb{u},\pmb{v})\in \mathcal{Z}$ maximizing the right-hand side of~(\ref{zuv}). In Section~\ref{secsp} we will show that the problem of computing $Z(\pmb{x})$ is NP-hard, when $\mathcal{P}$ is the \textsc{Shortest Path} problem. 
Using equality~(\ref{eqprop1}) from Proposition~\ref{propwc1},
we can compute the maximum regret of $\pmb{x}$ in the following way:
\begin{align*}
Z(\pmb{x})=\max\ &  \pmb{C}^T\pmb{x}-\pmb{C}^T\pmb{u}-\underline{\pmb{c}}^T\pmb{v}+z\\
\text{s.t. } & z \le  \pmb{c}_{\pmb{v}}^{T}\pmb{y} & \forall \pmb{y} \in \mathcal{R}(\pmb{x})  \\
	     & (\pmb{u}, \pmb{v})\in\mathcal{Z}
\end{align*}
which, by using the definition of $\pmb{c}_{\pmb{v}}$, can be stated equivalently as
\begin{align}
\max\ &  \pmb{C}^T\pmb{x}-\pmb{C}^T\pmb{u}-\underline{\pmb{c}}^T\pmb{v}+z \label{c08} \\
\text{s.t. } & z \le  \sum_{i\in [n]}(\underline{c}_iv_i + \overline{c}_i(1-v_i))y_i  & \forall \pmb{y} \in \mathcal{R}(\pmb{x})  \label{c09} \\
	     & (\pmb{u}, \pmb{v})\in\mathcal{Z} \label{c10}
\end{align}
The number of constraints in (\ref{c08})-(\ref{c10})  can be exponential in $n$.  This problem can be solved by standard row generation techniques. If the problem of optimizing a linear objective function over $\mathcal{R}(\pmb{x})$ can be written as a linear program, then it is also possible to find a compact reformulation of constraints~(\ref{c09}) 
using primal-dual relationships (see, e.g.,~\cite{PS98}). One such a case will be demonstrated in Section~\ref{secsel}. 
Let us now turn to the \textsc{TStR}~$\mathcal{P}$ problem. Again, using equality~(\ref{zuv}) we can express this problem as the following program:
\begin{align}
\min\ &  z \\
\text{s.t. } & z \ge  Z_{(\pmb{u},\pmb{v})}(\pmb{x}) & \forall (\pmb{u}, \pmb{v})\in\mathcal{Z} \label{c06} \\
	      & \pmb{x}\in \mathcal{X}'
\end{align}
Using Proposition~\ref{propwc1}, we can convert this model to
\begin{align*}
\min\ & \pmb{C}^T\pmb{x} + z \\
\text{s.t. } & z \ge \pmb{c}_{\pmb{v}}^{T}\pmb{y}^{\pmb{v}}- \pmb{C}^T\pmb{u}-\underline{\pmb{c}}^T\pmb{v} & \forall (\pmb{u}, \pmb{v})\in\mathcal{Z} \\
& \pmb{y}^{\pmb{v}} \in \mathcal{R}(\pmb{x}) &  \forall (\pmb{u}, \pmb{v})\in\mathcal{Z}\\
& \pmb{x}\in \mathcal{X}'
\end{align*}
Finally, making use of the definition of $\pmb{c}_{\pmb{v}}$ and $\mathcal{R}(\pmb{x})$, we get the following MIP formulation for \textsc{TStR}~$\mathcal{P}$:
\begin{align}
\min\ & \pmb{C}^T\pmb{x} + z \label{c12}\\
\text{s.t. } & z \ge \sum_{i\in [n]}(\underline{c}_iv_i + \overline{c}_i(1-v_i))y^{\pmb{v}}_i - \pmb{C}^T\pmb{u}-\underline{\pmb{c}}^T\pmb{v} & \forall (\pmb{u}, \pmb{v})\in\mathcal{Z} \\
& \pmb{y}^{\pmb{v}} \in \mathcal{R}(\pmb{x}) &  \forall (\pmb{u}, \pmb{v})\in\mathcal{Z}\\
& \pmb{x}\in \mathcal{X}' \label{c13}
\end{align}
The model~(\ref{c12})-(\ref{c13}) has an exponential number of variables and constraints. One can solve or approximate  it by using a row and column generation technique (see, e.g.,~\cite{ZZ13}). The idea is to solve~(\ref{c12})-(\ref{c13})  for some subset $\mathcal{Z}'\subseteq \mathcal{Z}$ obtaining a solution $\pmb{x}'\in \mathcal{X}'$,  together with a lower bound on the optimal objective value. The upper bound and the cut $(\pmb{u}',\pmb{v}')\in \mathcal{Z}$, which can be added to $\mathcal{Z}'$, can by computed by solving the formulation~(\ref{c08})-(\ref{c10}) for $\pmb{x}'$. Adding the cuts iteratively we can compute an exact or approximate solution to~\textsc{TStR}~$\mathcal{P}$. The efficiency of this method can depend on the structure of $\mathcal{X}$ and should be verified experimentally for each particular case.

\section{The shortest path problem}
\label{secsp}

In this section we deal with the case of $\textsc{TStR}~\mathcal{P}$, in which $\mathcal{P}$ is the \textsc{Shortest Path} problem. Let $G=(V,A)$ be a given network with two distinguished nodes $s\in V$ and $t\in V$. 
We will discuss two variants of the problem. In the first one, $\mathcal{X}_{\mathcal{P}}$ contains the characteristic vectors of all simple $s-t$ paths in $G$. In the second case, $\overline{\mathcal{X}}_{\mathcal{P}}$ contains all subsets of the arcs in which $s$ and $t$ are connected. Observe that $\mathcal{X}_{\mathcal{P}}\subseteq \overline{\mathcal{X}}_{\mathcal{P}}$ and in the deterministic case the problems with both sets are equivalent.
To see that the situation is different in the two-stage model, consider the sample instance shown in Figure~\ref{fig2}.
\begin{figure}[ht]
	\centering
	\includegraphics[height=4cm]{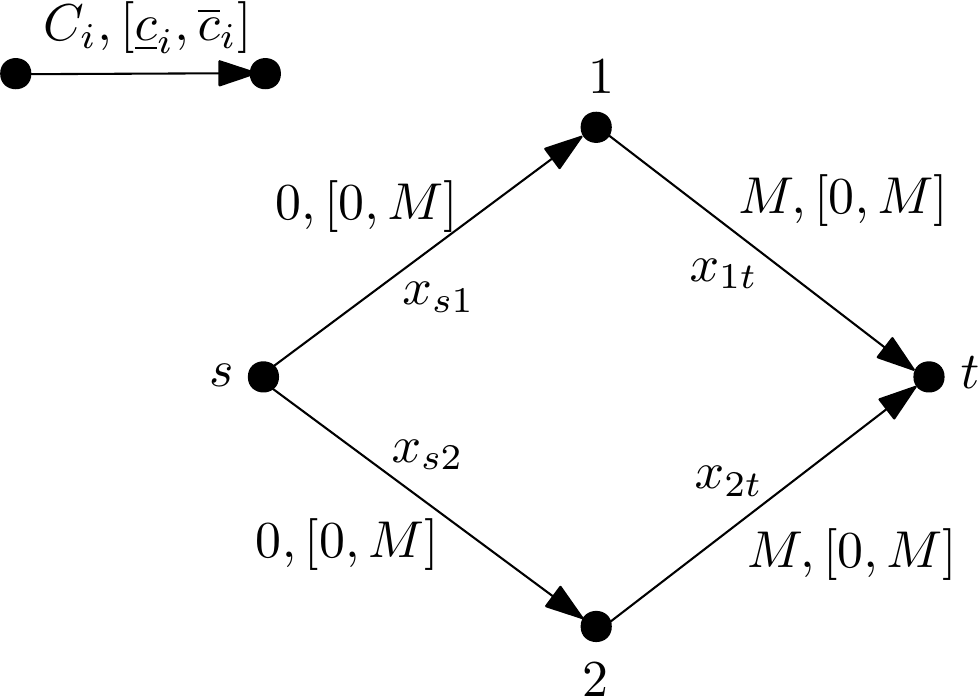}
	\caption{An instance of \textsc{TStR Shortest Path},
	where $M$ is a sufficiently large number.}\label{fig2}
\end{figure}
In the problem with set $\mathcal{X}_{\mathcal{P}}$ the maximum regret of each first-stage solution equals $M$. On the other hand, for the set $\overline{\mathcal{X}}_{\mathcal{P}}$ we can choose $\pmb{x}=(1,1,0,0)$, i.e. we can select the arcs $(s,1)$ and $(s,2)$ in the first stage. Then, depending on the second-stage scenario we can complete this solution by choosing $(1,t)$ or $(2,t)$. The maximum regret of $\pmb{x}$ is then~0. The example demonstrates that it can be profitable to select more arcs in the first stage, even if some of them are not ultimately used. The next theorem describes the 
computational complexity of \textsc{TStR Shortest Path}.

\begin{thm}\label{th:sphard}
The \textsc{TStR Shortest Path} problem with both $\mathcal{X}_{\mathcal{P}}$ and $\overline{\mathcal{X}}_{\mathcal{P}}$  is NP-hard.
\end{thm}

\begin{proof}

Consider the NP-hard \textsc{Partition} problem defined as follows~\cite{GJ79}. We are given a collection $a_1,\dots,a_n$ of positive integers, such that $\sum_{i\in [n]} a_i=2b$. We ask if there is a subset $I\subseteq [n]$ such that $\sum_{i\in I} a_i=b$. Given an instance of \textsc{Partition} we build the graph shown in Figure~\ref{fig1}. For each arc we specify the first-stage cost and the second-stage cost interval, where $M>2nb+2b$ is a sufficiently large constant (see Figure~\ref{fig1}).  We show that the answer to \textsc{Partition} is yes if and only if there is a first-stage solution $\pmb{x}$ with the maximum regret at most $\frac{3}{2}b$.  We first focus on the set $\mathcal{X}_{\mathcal{P}}$.
To this end, we prove the following three claims:

\begin{figure}[htb]
\centering
\includegraphics[width=\textwidth]{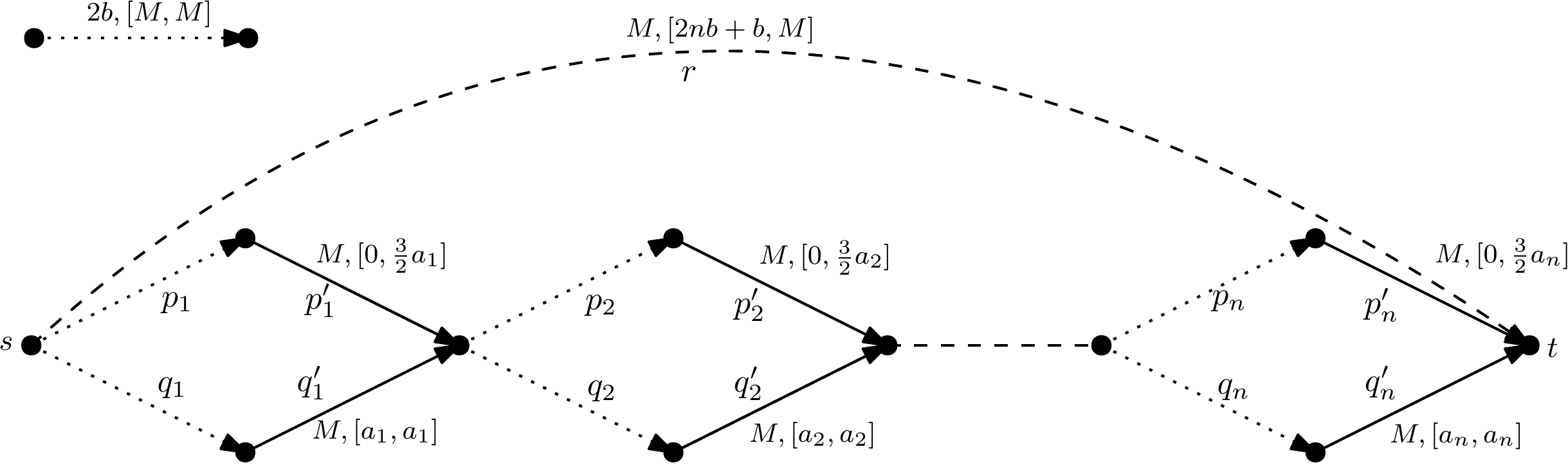}
\caption{The instance of \textsc{TStR Shortest Path} from the proof of Theorem~\ref{th:sphard}.}\label{fig1}
\end{figure}

\begin{enumerate}
\item \emph{There is an optimal first-stage  solution with the maximum regret at most $2b$.} Let $\pmb{x}$ be the first-stage solution in which all arcs $q_1,\dots,q_n$ are selected. Under any scenario, the optimal recourse action for $\pmb{x}$ selects the arcs $q_1',\dots,q_n'$. In the worst-case scenario $\underline{\pmb{c}}$ the second-stage costs of all arcs are set to their lower bounds. We get ${\rm Inc}(\pmb{x},\underline{\pmb{c}})=2nb+2b$,  ${\rm Opt}(\underline{\pmb{c}})=2nb$ and $Z(\pmb{x})=2b$.

\item \emph{The first-stage solution $\pmb{x}$, in which no arc is chosen in the first stage or at least one of the arcs among $r, p_i', q_i'$, $i\in [n]$, is chosen in the first stage is not optimal.} Define scenario $\pmb{c}_1$ in which the second-stage cost of arc $r$ is $M$ and the second-stage costs of  the arcs $p_1',\dots,p_n'$ are 0.  Observe that ${\rm Opt}(\pmb{c}_1)=2nb$. It is easy to see that ${\rm Inc}(\pmb{x},\pmb{c}_1)\geq M$.  Hence $Z(\pmb{x})\geq M-2nb>2b$ and $\pmb{x}$ is not optimal, according to point 1.

\item \emph{Any optimal first-stage solution $\pmb{x}$ selects exactly one of $p_i$ or $q_i$ for each $i\in [n]$}.  According to point 2. at least one of the arcs among $p_i, q_i$, $i\in [n]$, must be selected by $\pmb{x}$. Assume there is $k\in [n]$ such that both $p_k$ and $q_k$ are not selected in $\pmb{x}$. Consequently, we must choose $p_k$ or $q_k$ in the second stage which implies ${\rm Inc}(\pmb{x},\underline{\pmb{c}})\geq M$. Since ${\rm Opt}(\underline{\pmb{c}})=2nb$, we get $Z(\pmb{x})\geq M-2nb>2b$ and $\pmb{x}$ is not optimal, according to point 1.
  \end{enumerate}

Let $I\subseteq [n]$ be the set of indices of the arcs $p_i$ selected in the first stage. Notice that $[n]\setminus I$ is the set of indices of the arcs $q_i$ selected in the first stage. If arc $p_i$ is chosen in the first stage, then $p'_i$ must be chosen in the second state (the same is true for arcs $q_i$).  In the worst-case scenario $\pmb{c}$ we fix the cost of $r$ to $2nb+b$ and the cost of $p_i'$ to $\frac{3}{2}a_i$ if $i\in I$ and to~0, otherwise. We get ${\rm Inc}(\pmb{x}, \pmb{c})=2nb+\sum_{i\in I} \frac{3}{2}a_i+\sum_{i\in [n]\setminus I} a_i$  and ${\rm Opt}(\pmb{c})=2nb+\min\{b, \sum_{i\in I} a_i\}$. The maximum regret of the formed path $\pmb{x}$ is then
\begin{align*}
Z(\pmb{x})=&\max\left\{\sum_{i\in [n]\setminus I} a_i+\sum_{i\in I} \frac{3}{2}a_i-\sum_{i\in I} a_i, \sum_{i\in [n]\setminus I} a_i+\sum_{i\in I} \frac{3}{2}a_i-b\right\} \\
=&\max\left\{2b-\frac{1}{2}\sum_{i\in I} a_i, b+\frac{1}{2}\sum_{i\in I} a_i \right\}.
\end{align*}
We now can see that the maximum regret  of the formed path is at most $\frac{3}{2}b$ if and only if 
$\frac{1}{2}\sum_{i\in I} a_i=\frac{1}{2}b$, i.e. the answer to \textsc{Partition} is yes. 

Let us now turn to the case of $\overline{\mathcal{X}}_{\mathcal{P}}$. Because $\mathcal{X}_{\mathcal{P}}\subseteq \overline{\mathcal{X}}_{\mathcal{P}}$ the positive answer to the \textsc{Partition} problem implies that there is a first-stage solution $\pmb{x}\in \overline{\mathcal{X}}_{\mathcal{P}}$ such that $Z(\pmb{x})\leq \frac{3}{2}b$. It remains to show the converse implication, i.e. if there is a first-stage solution $\pmb{x}\in \overline{\mathcal{X}}_{\mathcal{P}}$ such that $Z(\pmb{x})\leq \frac{3}{2}b$, then the answer to \textsc{Partition} is yes. In the case of $\overline{\mathcal{X}}_{\mathcal{P}}$ any subset of arcs in the formed network $G$ is allowed to be selected in the first-stage. Similarly to the previous case, selecting any arc among $r, p_i', q_i'$, $i\in [n]$,  in the first stage yields a solution $\pmb{x}$ such that $Z(\pmb{x})>2b$. Hence, it may only be profitable to choose more than $n$ arcs among $p_i, q_i$, $i\in [n]$ in the first stage. Let $\pmb{x}$ be any such a solution. Under scenario $\underline{\pmb{c}}$, we get ${\rm Inc}(\pmb{x},\underline{\pmb{c}})\geq 2nb+2b$, while ${\rm Opt}(\underline{\pmb{c}})=2nb$. Hence $Z(\pmb{x})\geq 2b$, a contradiction with the assumption that $Z(\pmb{x})\leq \frac{3}{2}b$.

\end{proof}

 \begin{thm}\label{th:regr}
Computing the maximum regret $Z(\pmb{x})$ for a given $\pmb{x}\in \mathcal{X}'$ is NP-hard for both $\mathcal{X}_{\mathcal{P}}$ and $\overline{\mathcal{X}}_{\mathcal{P}}$.
\end{thm}
\begin{proof}
Given again an instance of the \textsc{Partition} problem (see the proof of Theorem~\ref{th:sphard}).
We construct a network $G=(V,A)$ consisting of two disjoint paths from $s$ to $t$, $P_1$ and $P_2$, each with
 with $n$ arcs. The network $G$ with the corresponding first-stage costs and the second-stage cost intervals is shown in Figure~\ref{fig:eval}.
 Set $\pmb{x}=\pmb{0}$, so no arc is allowed to be selected in the first stage. We will show that the answer \textsc{Partition} is yes if and only if $Z(\pmb{x})\geq b$.

\begin{figure}[htb]
\begin{center}
\includegraphics[width=.8\textwidth]{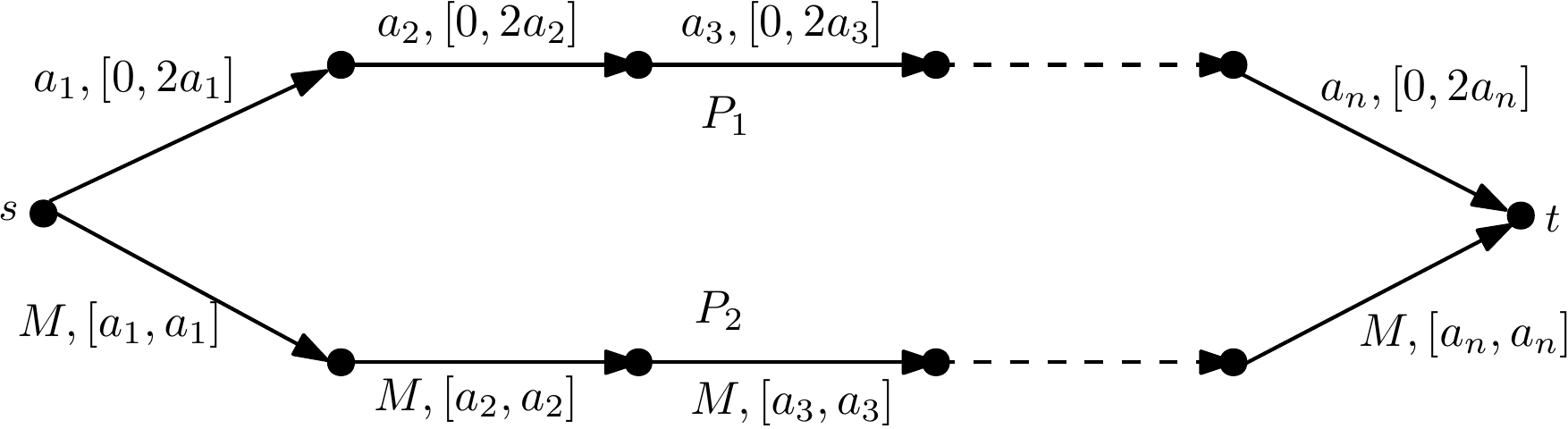}
\caption{The instance of \textsc{TStR Shortest Path} from the proof of Theorem~\ref{th:regr}.}\label{fig:eval}
\end{center}
\end{figure}

Let $\pmb{c}$ be a worst-scenario for $\pmb{x}$. According to Corollary~\ref{wcsc}, we can assume that $\pmb{c}$ is an extreme scenario. Let $I$ be the set of indices of the arcs in $P_1$ whose second-stage costs are set to the upper bounds. The value of  ${\rm Inc}(\pmb{x},\pmb{c})$ is the minimum of $\sum_{i\in [n]} a_i$ (path $P_2$ is used as the best recourse action) and $\sum_{i\in I} 2a_i$ (path $P_1$ is used as the best recourse action). On the other hand ${\rm Opt}(\pmb{c})=\sum_{i\in I} a_i$, because the optimal two-stage path is $P_1$, where the arcs with indices in $I$ are selected in the first-stage. Hence
$$Z(\pmb{x})=\min\left\{\sum_{i\in [n]} a_i, \sum_{i\in I} 2a_i\right\}-\sum_{i\in I} a_i=\min\left\{\sum_{i\in [n]\setminus I} a_i, \sum_{i\in [n]} a_i\right\}.$$
Therefore, $Z(\pmb{x})\geq b$ if and only if the answer to \textsc{Partition} problem is yes.
\end{proof}

In the following we analyze the computational complexity of the incremental and two-stage variants of the problem. 

\begin{obs}
The \textsc{Inc Shortest Path} problem with $\overline{\X}_{\mathcal{P}}$ can be solved in polynomial time.
\end{obs}
\begin{proof}
Given network $G=(V,A)$, a first-stage solution $\pmb{x}$ and a second-stage cost scenario $\pmb{c}$, consider a deterministic \textsc{Shortest Path} problem in the same network $G$, in which the costs of the arcs selected in $\pmb{x}$ are~0 and the costs of the remaining arcs are determined according to $\pmb{c}$. We seek a shortest $s-t$ path in $G$. The optimal recourse action $\pmb{y}\in \mathcal{R}(\pmb{x})$ selects the arcs on the path computed, which are not selected by $\pmb{x}$. Notice that $\pmb{x}+\pmb{y}$ needs not to describe a simple $s-t$ path in $G$.
\end{proof}

\begin{obs}
	The \textsc{TSt Shortest Path} problem with both $\overline{\X}_{\mathcal{P}}$ and $\X_{\mathcal{P}}$ can be solved in polynomial time.
\end{obs}
\begin{proof}
Given a network $G=(V,A)$  in which $C_i$ is the first-stage cost and $c_i$ is the second-stage cost of the arc $a_i\in A$, consider a deterministic \textsc{Shortest Path} problem in the same network $G$, in which the cost of the arc $a_i$ is $\hat{c}_i=\min\{C_i, c_i\}$. Let $P$ be the shortest $s-t$ path in $G$ with the arc costs $\hat{c}_i$. We form  solution $(\pmb{x},\pmb{y})\in \mathcal{Z}$ as follows: for each $a_i\in P$, if $\hat{c}_i=C_i$, then $x_i=1$; if $\hat{c}_i<C_i$, then $y_i=1$; for each $a_i\notin P$, $x_i=y_i=0$. One can easily verify that $(\pmb{x},\pmb{y})$ is an optimal solution to \textsc{TSt Shortest Path} regardless of which set $\overline{\X}_{\mathcal{P}}$ or $\X_{\mathcal{P}}$ is used.
\end{proof}

The next result shows a difference between the two problem variants using $\X_{\mathcal{P}}$ and $\overline{\X}_{\mathcal{P}}$.

\begin{thm}\label{th:sp-inner}
The \textsc{Inc Shortest Path} problem with $\X_P$ is strongly NP-hard and not at all approximable if P$\neq$NP.
\end{thm}
\begin{proof}
Consider the following strongly NP-complete \textsc{Hamiltonian Path} problem~\cite{GJ79}. We are given a directed graph $G=(V,A)$ with two distinguished nodes $v_1,v_n\in V$. We ask if it is possible to find a directed path from $v_1$ to $v_n$ that visits each node in $V$ exactly once. We build a graph $G'=(V',A')$ as follows. For each $v_i\in V$, the node set $V'$ contains two nodes $v_i$ and $v_i'$. The set of  arcs $A'$ contains the \emph{forward arcs} $(v_i,v_i')$ for all $v_i\in V$, and \emph{backward arcs} $(v_i',v_j)$ for all $(v_i,v_j)\in A$. Finally, there are \emph{dummy arcs} $(v_i',v_{i+1})$ for all $i\in[n-1]$. The first stage costs of all arcs in $A'$ are~0. Under the second-stage scenario $\pmb{c}$, the costs of all forward and backward arcs are~0 and the costs of all dummy arcs are~1. We set $s=v_1$ and $t=v_n$. Figure~\ref{fig:sp-inner} shows a sample reduction. 

\begin{figure}[htb]
\begin{center}
\includegraphics[width=0.5\textwidth]{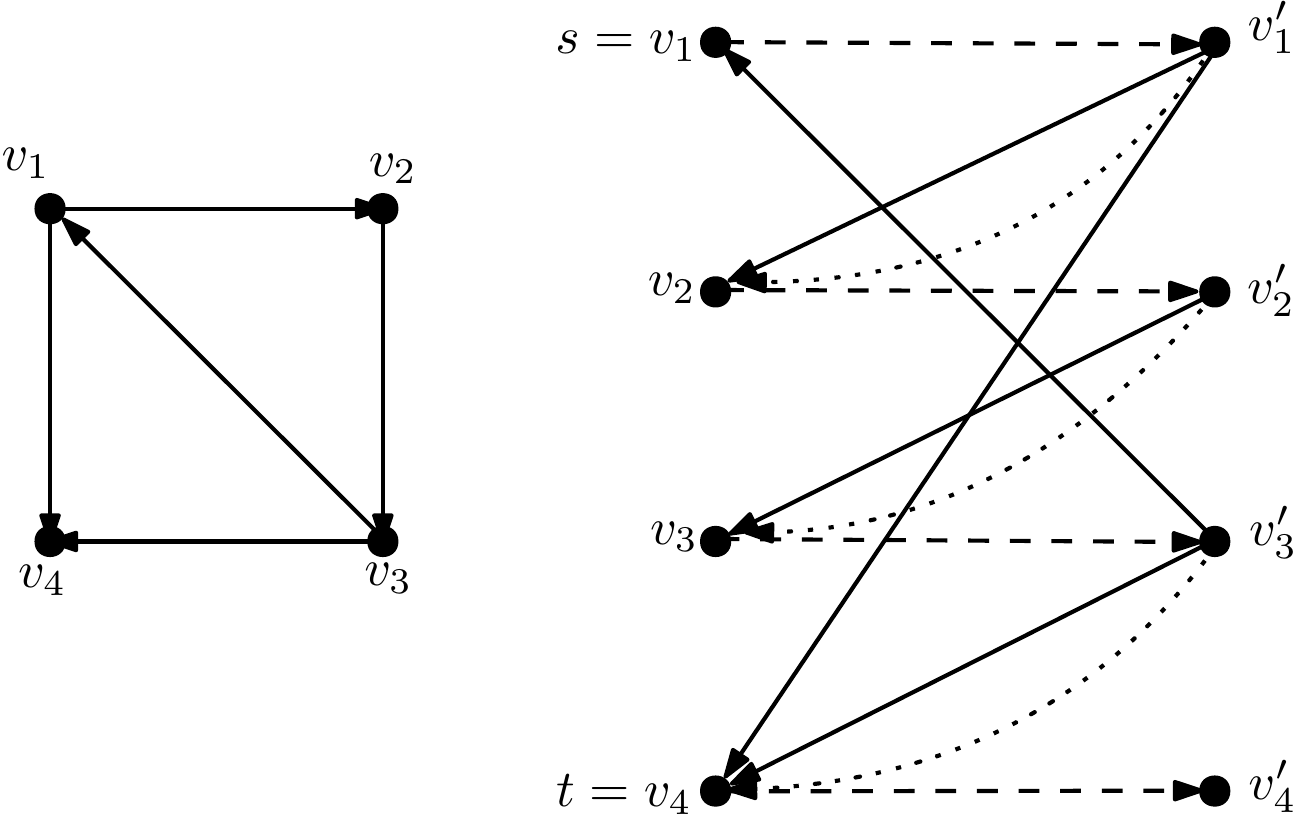}%
\caption{A sample reduction in the proof of Theorem~\ref{th:sp-inner}. Arcs with zero costs are solid. Dashed forward arcs were chosen in the first stage.}\label{fig:sp-inner}
\end{center}
\end{figure}

The presence of dummy arcs ensures that $\pmb{x}$ is feasible, as we can complete $\pmb{x}$ to a simple path by using all dummy arcs. Now it is easy to see that there is an optimal recourse action $\pmb{y}\in \mathcal{R}(\pmb{x})$, which selects only the backward arcs  if and only if there is a Hamiltonian path in $G$. In other words, there is $\pmb{y}\in \mathcal{R}(\pmb{x})$ such that ${\rm Inc}(\pmb{x},\pmb{c})=0$ if and only if $G$ has a Hamiltonian path, which proves the theorem.

\end{proof}

\section{The selection problem}
\label{secsel}

In this section we discuss the  \textsc{TStR Selection} problem, in which $\X_{\mathcal{S}}=\{\pmb{x}\in \{0,1\}^n: x_1+\dots+x_n=p\}$ for some fixed $p\in [n]$. We can interpret $\pmb{x}\in \X_{\mathcal{S}}$ as a characteristic vector of a selection of exactly $p$ items out of $n$ available. The set of feasible first-stage solutions is then $\X'_{\mathcal{S}}=\{\pmb{x}\in \{0,1\}^n: x_1+\dots+x_n\leq p\}$. The deterministic \textsc{Selection} problem can be solved in $O(n)$ time. We first find in $O(n)$ time the $p$th smallest item cost $c$ (see, e.g.,~\cite{CO90})
and select $p$ items with the costs at most $c$. It is easy to see that the corresponding \textsc{TSt Selection} and \textsc{Inc Selection} problems are solvable in $O(n)$ time as well. Various robust versions of \textsc{Selection} have been discussed in the literature. In particular, the single-stage minmax and minmax regret models were investigated in~\cite{A01, C04, D13, KKZ13} and the robust two-stage models, with the minmax criterion, were discussed in~\cite{CGKZ18, KZ15b}.

\subsection{Computing the maximum regret}
\label{secregrsel}

In this section we show that the value of $Z(\pmb{x})$ for a given $\pmb{x}\in \mathcal{X}'_{\mathcal{S}}$ can be computed in polynomial time. Consider the subproblem $\min_{\pmb{y}\in \mathcal{R}(\pmb{x})} \pmb{c}_{\pmb{v}}^T\pmb{y}$ from Proposition~\ref{propwc1}, which can be represented as the following linear programming problem:
 \begin{align}
\min\ &  \pmb{c}_{\pmb{v}}^T\pmb{y} \label{con-in0} \\
\text{s.t. } & \sum_{i\in[n]} y_i \ge p - \sum_{i\in[n]} x_i \label{con-in1}\\
& y_i \le 1 - x_i & \forall i\in[n] \label{con-in2} \\
& y_i\geq 0 & \forall i\in[n] \label{con-in3}
\end{align}
Observe that we have relaxed the constraints $y_i\in \{0,1\}$, $i\in [n]$, without changing the optimal objective function value, because the problem (\ref{con-in0})-(\ref{con-in3}) has an integral optimal solution, which is due to total unimodularity
of the constraint matrix (\ref{con-in1})-(\ref{con-in2}).
Using the definition of $\pmb{c}_{\pmb{v}}$ (see Proposition~\ref{propwc1}), we can write the dual to~(\ref{con-in0})-(\ref{con-in3}), which is another linear programming problem of the form:
 \begin{align*}
\max\ & (p-\sum_{i\in[n]}x_i)\alpha + \sum_{i\in[n]} (x_i-1) \beta_i \\
\text{s.t. } & \alpha - \beta_i \le \underline{c}_iv_i + \overline{c}_i(1-v_i) & \forall i\in[n] \\
& \alpha, \beta_i \ge 0 & \forall i\in[n]
\end{align*}
Using Proposition~\ref{propwc1} and equality~(\ref{zuv}), we can build the following compact mixed integer programming formulation for computing the maximum regret of a given first-stage solution $\pmb{x}\in \mathcal{X}'_{\mathcal{S}}$:
\begin{align}
Z(\pmb{x})=\max\ & \pmb{C}^T\pmb{x}-\pmb{C}^T\pmb{u}-  \underline{\pmb{c}}^T\pmb{v}+(p-\sum_{i\in[n]}x_i)\alpha + \sum_{i\in[n]} (x_i-1) \beta_i   \label{mipz0}\\
\text{s.t. } & \alpha - \beta_i \le \underline{c}_iv_i + \overline{c}_i(1-v_i)& \forall i\in[n] \\
& \sum_{i\in[n]} (u_i + v_i) = p \label{mipz2}\\
& u_i + v_i \le 1 & \forall i\in[n] \label{mipz3}\\
& u_i,v_i \in \{0,1\} & \forall i\in[n] \label{mipz4}\\
& \alpha, \beta_i \ge 0 & \forall i\in[n]  \label{mipz8}
\end{align}
The constraints (\ref{mipz2})-(\ref{mipz4}) represent a feasible pair $(\pmb{u},\pmb{v})\in \mathcal{Z}$.
Define $X=\{i\in [n]: x_i=1\}$ and let $\overline{X}=\{i\in [n]: x_i=0\}$ be the complement of $X$.  Observe that in an optimal solution to~(\ref{mipz0})-(\ref{mipz8}) we can fix $\beta_i=[\alpha- \underline{c}_iv_i - \overline{c}_i(1-v_i)]_{+}$ for each $i\in [n]$ (we use the notation $[t]_{+}=\max\{0,t\}$). As $v_i\in\{0,1\}$, we can set:
\begin{equation}
\label{z00}
\sum_{i\in [n]}(x_i-1)\beta_i=- \sum_{i\in \overline{X}} ([\alpha-\underline{c}_i]_+  - [\alpha-\overline{c}_i]_+) v_i -\sum_{i\in \overline{X}} [\alpha-\overline{c}_i]_{+}.
\end{equation}
Denote $c^*_i(\alpha)=[\alpha-\underline{c}_i]_+  - [\alpha-\overline{c}_i]_+$ and, using~(\ref{z00}), rewrite (\ref{mipz0})-(\ref{mipz8}) as 

\begin{align}
Z(\pmb{x})=\max\ & \pmb{C}^T\pmb{x}-\pmb{C}^T\pmb{u}-  \underline{\pmb{c}}^T\pmb{v}+(p-|X|)\alpha - \sum_{i\in \overline{X}} c^*_i(\alpha)v_i -\sum_{i\in \overline{X}} [\alpha-\overline{c}_i]_{+}   \label{mipz10}\\
\text{s.t. }& \sum_{i\in[n]} (u_i + v_i) = p \label{mipz12}\\
& u_i + v_i \le 1 & \forall i\in[n] \label{mipz13}\\
& u_i,v_i \in \{0,1\} & \forall i\in[n] \label{mipz14}\\
& \alpha \ge 0 &  \label{mipz18}
\end{align}
Define $\mathcal{C}=\{\underline{c}_1,\dots,\underline{c}_n\}\cup\{\overline{c}_1,\dots,\overline{c}_n\}$.
\begin{lem}
\label{lemz1}
	There is an optimal solution to~(\ref{mipz10})-(\ref{mipz18}) in which $\alpha\in \mathcal{C}$.
\end{lem}
\begin{proof}
	  Fix any $(\pmb{u},\pmb{v})\in \mathcal{Z}$ to~(\ref{mipz10})-(\ref{mipz18}). The optimal value of $\alpha\geq 0$ can be then found by solving the following problem:
\begin{equation}
\label{falpha}
\max_{\alpha\geq 0}\left((p-|X|) \alpha-\sum_{i\in \overline{X}} [\alpha-\hat{c}_i]_{+}\right),
\end{equation}
where $\hat{c}_i=\underline{c}_i$ if $v_i=1$ and $\hat{c}_i=\overline{c}_i$, otherwise.  Since $|\overline{X}|\geq p-|X|$, the problem~(\ref{falpha}) attains a maximum at some $\alpha\geq 0$. As the objective function of this problem is piecewise linear,  the optimal value of $\alpha$ is at some $\hat{c}_k$, $k\in \overline{X}$. Since $\hat{c}_k\in \mathcal{C}$, the lemma follows.
\end{proof}

\begin{thm}
	The value of $Z(\pmb{x})$ for a given $\pmb{x}\in \X_{\mathcal{S}}'$ can be computed in $O(n^2)$ time.
\end{thm}
\begin{proof}
	Fix $\alpha\in \mathcal{C}$  in (\ref{mipz10})-(\ref{mipz18}). The remaining optimization problem is then
\begin{align}
\min &\; \pmb{C}^T\pmb{u}+\hat{\pmb{c}}^T\pmb{v} 
\label{sub-start}  \\
\text{s.t. } & \sum_{i\in[n]} (u_i + v_i) = p \\
& u_i + v_i \le 1 & \forall i\in [n] \\
& u_i,v_i \in \{0,1\} & \forall i\in[n], \label{sub-end}
\end{align}
where $\hat{c}_i=\underline{c}_i$ if $i\notin \overline{X}$ and $\hat{c}_i=\underline{c}_i+[\alpha-\underline{c}_i]_{+}-[\alpha-\overline{c}_i]_{+}$, otherwise. 
Observe that (\ref{sub-start})-(\ref{sub-end}) is a \textsc{TSt Selection} problem with the first-stage costs $\pmb{C}$ and the second-stage costs $\hat{\pmb{c}}$, which can be solved in $O(n)$ time. By Lemma~\ref{lemz1} it is enough to try at most $2n$ values of $\alpha$ in $\mathcal{C}$ to find an optimal solution to~(\ref{mipz10})-(\ref{mipz18}). Therefore, one can solve this problem and thus compute $Z(\pmb{x})$ in $O(n^2)$ time.
	
\end{proof}

\subsection{Compact MIP formulations}

Fix $\alpha \in \mathcal{C}$ in~(\ref{mipz10})-(\ref{mipz18}). One can easily check that the constraint matrix (\ref{mipz12})-(\ref{mipz13}) is totally unimodular. Hence we can relax the constraints $u_i, v_i\in \{0,1\}$ with $u_i, v_i \ge 0$ and write the following dual to  the relaxed (\ref{mipz10})-(\ref{mipz18}):

\begin{align*}
\min\ & \pmb{C}^T\pmb{x}+ (p-|X|)\alpha-\sum_{i\in [n]}[\alpha-\overline{c}_i]_{+}(1-x_i)+ p \pi(\alpha) + \sum_{i\in[n]} \rho_i(\alpha) \\
\text{s.t. } & \pi(\alpha) + \rho_i(\alpha) \ge -C_i &\forall i\in[n]\\
& \pi(\alpha) + \rho_i(\alpha) \ge -\underline{c}_i-c^*_i(\alpha)(1-x_i)  & \forall i\in[n] \\
& \rho_i(\alpha) \ge 0 & \forall i\in[n]
\end{align*}
which can be rewritten as:
\begin{align*}
\min\ & \pmb{C}^T\pmb{x}+ (p-\sum_{i\in [n]} x_i)\alpha+\sum_{i\in [n]}[\alpha-\overline{c}_i]_{+}(x_i-1)- p \pi(\alpha) + \sum_{i\in[n]} \rho_i(\alpha)  \\
\text{s.t. } & \pi(\alpha) - \rho_i(\alpha) \le C_i &\forall i\in[n]\\
& \pi(\alpha) - \rho_i(\alpha) \le \underline{c}_i+c^*_i(\alpha)(1-x_i)  & \forall i\in[n] \\
& \rho_i(\alpha) \ge 0 & \forall i\in[n]
\end{align*}
We can thus construct the following MIP formulation for \textsc{TStR Selection}:
\begin{align}
\min\ &  \pmb{C}^T\pmb{x}+ z \label{mipsel0} \\
\text{s.t. } &  z\geq (p-\sum_{i\in [n]} x_i)\alpha+\sum_{i\in [n]}[\alpha-\overline{c}_i]_{+}(x_i-1)- p \pi(\alpha) + \sum_{i\in[n]} \rho_i(\alpha)  & \alpha\in \mathcal{C}\\
& \pi(\alpha) - \rho_i(\alpha) \le C_i &\forall i\in[n], \alpha\in \mathcal{C}\\
& \pi(\alpha) - \rho_i(\alpha) \le  \underline{c}_i+c^*_i(\alpha)(1-x_i) & \forall i\in[n], \alpha\in \mathcal{C} \\
& \sum_{i\in [n]} x_i\leq p \\
& \rho_i(\alpha) \ge 0 & \forall i\in[n], \alpha\in \mathcal{C} \\
& x_i\in \{0,1\} & i\in [n] \label{mipsel1}
\end{align}
In the following we will show how to decompose (\ref{mipsel0})-(\ref{mipsel1}) into a family of problems of smaller size. The key idea will be to show that it is enough to try at most $O(n^2)$ cases for the vector of variables $\pmb{\pi}=(\pi(\alpha))_{\alpha\in \mathcal{C}}$.
\begin{lem}
\label{lemenum}
	There exist $\hat{c}_k\in \mathcal{A}=\{C_1,\dots,C_n\}\cup\{\underline{c}_1,\dots,\underline{c}_n\}$ and $\hat{c}_l\in \mathcal{B}=\{C_1,\dots,C_n\} \cup\{\underline{c}_1,\dots,\underline{c}_n\}\cup \{\overline{c}_1,\dots,\overline{c}_n\}$, $
\hat{c}_k\leq \hat{c}_l$ such that
	\[ \pi(\alpha) = \max\{\hat{c}_k, \min\{\alpha, \hat{c}_l\}\}=\begin{cases}
\hat{c}_k & \text{ if } \alpha \le \hat{c}_k \\
\alpha & \text{ if } \hat{c}_k < \alpha <\hat{c}_l \\
\hat{c}_l & \text{ if } \hat{c}_l \le \alpha
\end{cases}, \; \alpha\in \mathcal{C} \]
are optimal to (\ref{mipsel0})-(\ref{mipsel1}).
\end{lem}
\begin{proof}
	Let us fix $\pmb{x}$ and $\alpha$ in (\ref{mipsel0})-(\ref{mipsel1}), and define $r_i(\alpha,x_i) = \min\{ C_i, \underline{c}_i+c_i^*(\alpha)(1-x_i)\}=\min\{C_i, \underline{c}_i+([\alpha-\underline{c}_i]_+  - [\alpha-\overline{c}_i]_+)(1-x_i)\}$. The values of  $\pi(\alpha)$ and $\rho_i(\alpha)$ can be then found by solving the following linear programming problem:
\begin{align*}
\max\ & p\pi(\alpha) - \sum_{i\in[n]} \rho_i(\alpha) \\
\text{s.t. } & \pi(\alpha) - \rho_i(\alpha) \le r_i(\alpha,x_i) & \forall i\in[n] \\
& \rho_i(\alpha) \ge 0 & \forall i\in[n] 
\end{align*}
By using linear programming duality, one can check that in an optimal solution to this problem, we can set $\pi(\alpha)$ to the $p$th smallest value among $r_i(\alpha,x_i)$, $i\in [n]$.
We now consider all possible shapes of $r_i(\alpha,x_i)$. \begin{enumerate}
\item $r_i(\alpha,1)= \min\{C_i,\underline{c}_i\}$,
\item If $C_i \le \underline{c}_i \le \overline{c}_i$, then  $r_i(\alpha,0)=C_i$.

\item If $\underline{c}_i \le C_i \le \overline{c}_i$, then
\[ r_i(\alpha,0) = \begin{cases} \underline{c}_i & \text{ if } \alpha \le \underline{c}_i \\
\alpha & \text{ if } \underline{c}_i < \alpha < C_i \\
C_i & \text{ if } C_i \le \alpha \end{cases} \]

\item If $\underline{c}_i \le \overline{c}_i \le C_i$, then
\[ r_i(\alpha,0) = \begin{cases} \underline{c}_i & \text{ if } \alpha \le \underline{c}_i \\
\alpha & \text{ if } \underline{c}_i < \alpha < \overline{c}_i \\
\overline{c}_i & \text{ if } \overline{c}_i \le u \end{cases} \]
\end{enumerate}
 The three cases for $r_i(\alpha,0)$ are visualized in Figure~\ref{fig-pi-cases}.
All possible shapes have in common that they have a constant value in $\{C_i,\underline{c}_i\}$ up to the diagonal. They then follow the diagonal to leave at another constant value in $\{C_i, \underline{c}_i, \overline{c}_i\}$. This means that the function representing the $p$th smallest value over all $r_i(u,x_i)$ is also of this shape, which gives the possibilities for $\pi(\alpha)$ as claimed.
\begin{figure}[htb]
\begin{center}
\subfigure[Case $C_i = 2 < \underline{c}_i$]{\includegraphics[width=.3\textwidth]{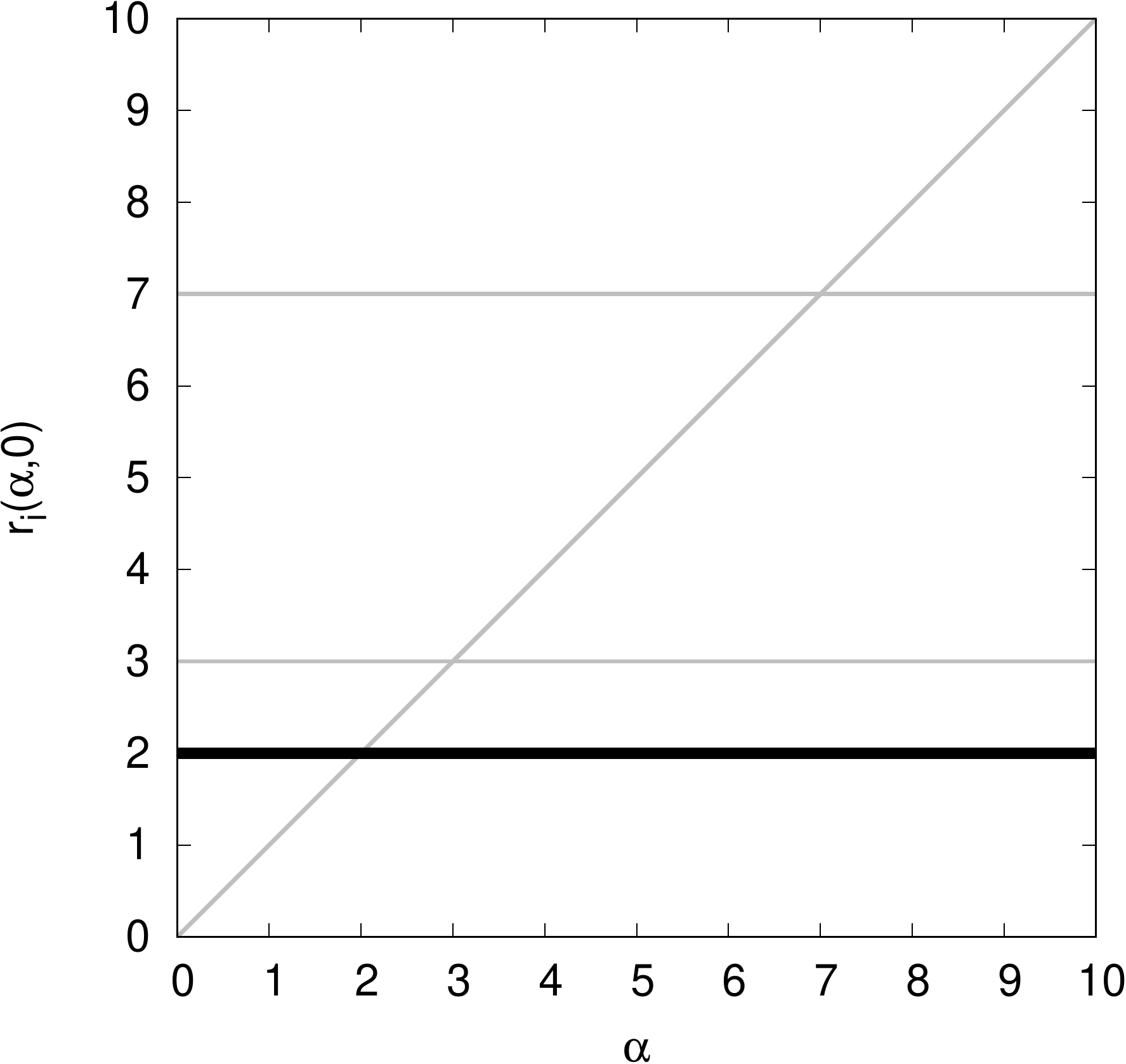}}%
\hfill
\subfigure[Case $\underline{c}_i < C_i = 5 < \overline{c}_i$]{\includegraphics[width=.3\textwidth]{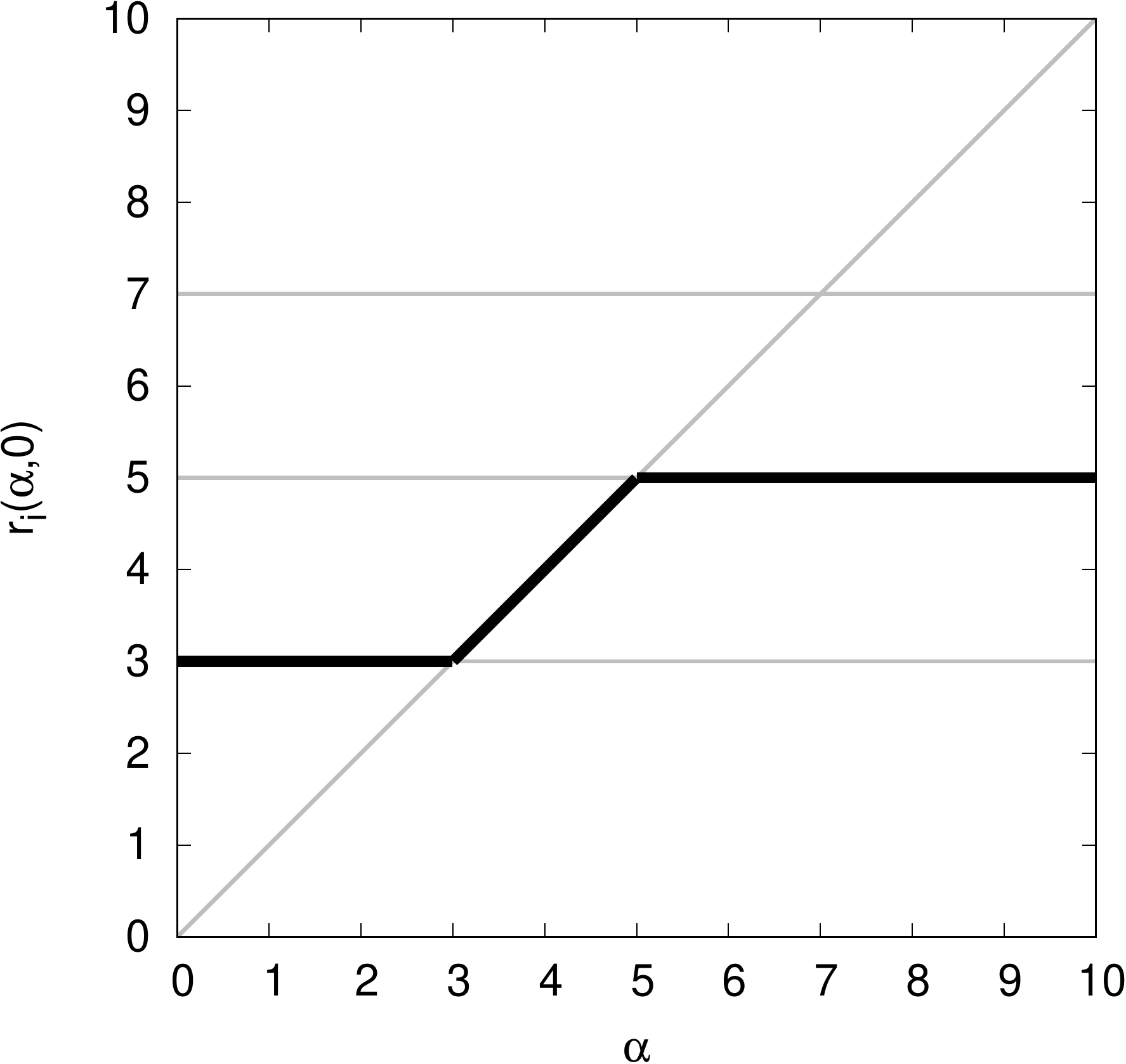}}%
\hfill
\subfigure[Case $\overline{c}_i < C_i = 8$]{\includegraphics[width=.3\textwidth]{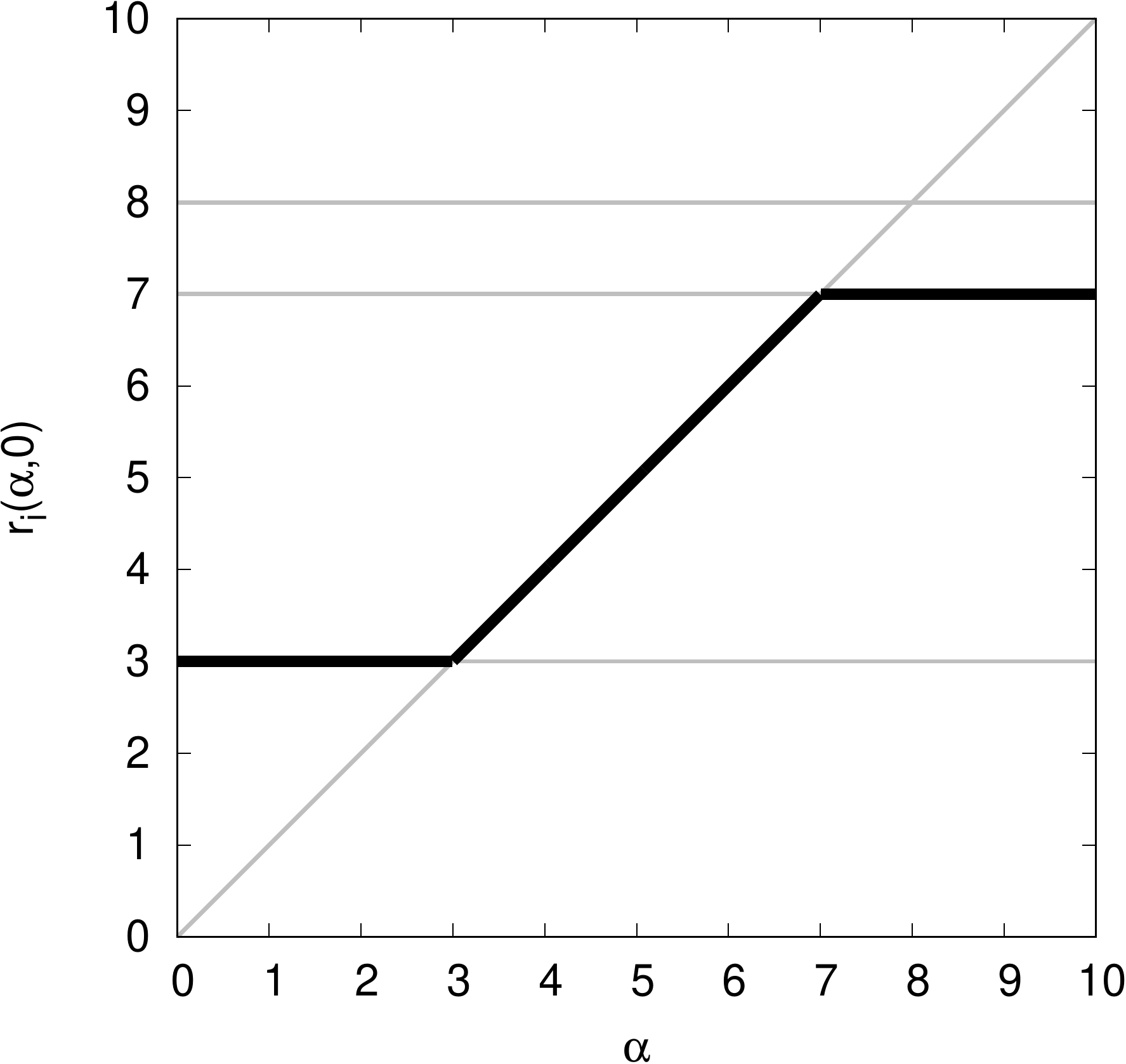}}
\caption{Shape of $r_i(\alpha,0)$ with $\underline{c}_i = 3$, $\overline{c}_i = 7$.}\label{fig-pi-cases}
\end{center}
\end{figure}
\end{proof}

Using Lemma~\ref{lemenum} we can enumerate all vectors $\pmb{\pi}=(\pi(\alpha))_{\alpha\in \mathcal{C}}$ and denote the set of these vectors by $\Pi$. Notice that $|\Pi|$ is $O(n^2)$.  For a fixed $\pmb{\pi}=(\pi(\alpha))_{\alpha\in \mathcal{C}}\in \Pi$, we can rewrite the problem (\ref{mipsel0})-(\ref{mipsel1}) as follows:
\begin{align*}
{\rm P}(\pmb{\pi}) = \min\ & z \\
\text{s.t. } & z \ge \nu^{\pmb{\pi}}(\alpha) + \sum_{i\in[n]} \omega^{\pmb{\pi}}_i(\alpha) x_i & \forall \alpha\in \mathcal{C} \\
& \sum_{i\in[n]} x_i \le p \\
& x_i \in \{0,1\} & \forall i\in [n]
\end{align*}
where
\begin{align*}
\nu^{\pmb{\pi}}(\alpha) &=  p\alpha - \sum_{i\in[n]} [\alpha-\overline{c}_i]_+ - p\pi(\alpha) + \sum_{i\in[n]} \underline{\rho}^{\pmb{\pi}}_i(\alpha) \\
\omega^{\pmb{\pi}}_i(\alpha) &= C_i - \alpha + [\alpha-\overline{c}_i]_+ + \overline{\rho}^{\pmb{\pi}}_i(\alpha) - \underline{\rho}^{\pmb{\pi}}_i(\alpha) 
\end{align*}
and
\begin{align*}
\underline{\rho}^{\pmb{\pi}}_i(\alpha) &= \max\{0,\pi(\alpha) - C_i, \pi(\alpha) - \underline{c}_i - [\alpha-\underline{c}_i]_+ + [\alpha-\overline{c}_i]_+ \} \\
\overline{\rho}^{\pmb{\pi}}_i(\alpha) &= \max\{ 0 , \pi(\alpha) - C_i, \pi(\alpha) - \underline{c}_i\}
\end{align*}
are constant values. We therefore find that
\begin{equation}
\label{eqzpi}
 \min_{\pmb{x} \in\X'} Z(\pmb{x}) = \min_{\pmb{\pi}\in\Pi} {\rm P}(\pmb{\pi}). 
 \end{equation}
According to Lemma~\ref{lemenum} we enumerate $O(n^2)$ many candidate vectors $\pmb{\pi}\in\Pi$. We then solve the resulting problem with $n$ binary variables and $O(n)$ constraints, which is substantially smaller than the MIP formulation (\ref{mipsel0})-(\ref{mipsel1}). In the next section we will propose a heuristic greedy algorithm, which is based on the decomposition~(\ref{eqzpi}).

The computational complexity of the general \textsc{TStR Selection} problem remains open. In the following we will identify some of its special cases which can be solved in polynomial time. 

\begin{prop}
\label{proppn}
	If $p=n$, then \textsc{TStR Selection} can be solved in polynomial time.
\end{prop}
\begin{proof}
	For each $\pmb{x}\in \mathcal{X}'$ and $\pmb{y}\in \mathcal{R}(\pmb{x})$, we get $\pmb{x}+\pmb{y}=n$.
	 Hence for each $i\in [n]$, if $x_i=1$, then the contribution of $i$ to the maximum regret is $C_i-\min\{C_i,\underline{c}_i\}$ and if $x_i=0$, then the contribution is $\overline{c}_i-\min\{C_i, \overline{c}_i\}$. Hence we set $x_i=1$ if $C_i-\min\{C_i,\underline{c}_i\}\leq \overline{c}_i-\min\{C_i, \overline{c}_i\}$ for every $i\in [n]$.
\end{proof}

\begin{prop}
The \textsc{TStR Selection} problem can be solved in polynomial time if two out of the following three conditions hold:
\begin{enumerate}
\item[(1)] The set $\{ C_i : i\in[n] \}$ is of constant size.
\item[(2)] The set $\{ \underline{c}_i : i\in[n] \}$ is of constant size.
\item[(3)] The set $\{ \overline{c}_i : i\in[n] \}$ is of constant size.
\end{enumerate}

\end{prop}
\begin{proof}
We first consider the case in which the conditions (1) and (2) hold. Let the sets $T_j$, $j\in [\ell]$, consist of all $i\in [n]$ with the same values $C_i$ and $\underline{c}_i$. Due to (1) and (2), $\ell$ is constant as well.  Let $(l_1,\dots,l_{\ell})$ be a vector of nonnegative integers such that $l_1+\dots+l_{\ell}=p$. 
This vector defines a decomposition into subproblems where we pick $l_j$ many items out of each set $T_j$. These items $i\in T_j$ only differ with respect to their upper bounds $\overline{c}_i$. An optimal solution is hence to pick those $l_j$ many items $i$ in the first stage that have the highest second-stage costs $\overline{c}_i$. The number of the vectors $\pmb{l}$ enumerated is $O(n^{\ell})$ so it remains polynomial. 
The cases that assumptions (1) and (3) as well as assumptions (2) and (3) hold can be treated in the same way.
\end{proof}

\subsection{Greedy algorithm}
In this section we propose a heuristic algorithm for computing a solution to \textsc{TStR Selection}, which can be applied to larger instances.  The first idea consists in applying the mid-point scenario heuristic, i.e. to solve \textsc{TSt Selection} under scenario $\pmb{c}^m$ such that $c^m_i=(\underline{c}_i+\overline{c}_i)/2$ for each $i\in [n]$. Unfortunately, the approximation ratio of this algorithm is unbounded, which can be easily demonstrated by using an instance analogous to that in Figure~\ref{figex1}~b. Observe that the \textsc{TStR Shortest Path} instance in this figure can be seen as an instance of \textsc{TStR Selection} with $n=2$ and $p=1$. We will now propose a more complex heuristic for the problem, which is based on equation~(\ref{eqzpi}). Given $X\subseteq [n]$, let us define
\[ F^{\pmb{\pi}}(X) = \max_{\alpha\in \mathcal{C}} \left( \nu^{\pmb{\pi}}(\alpha) + \sum_{i\in X} \omega^{\pmb{\pi}}_i(\alpha) \right). \]
Using~(\ref{eqzpi}), we get
$$\min_{\pmb{x}\in \mathcal{X}'}Z(\pmb{x})=\min_{\pmb{\pi}\in \Pi} \min_{\{X :X\subseteq [n],|X|\leq p\}} F^{\pmb{\pi}}(X). $$

\begin{thm} \label{thm-sup}
Function $F^{\pmb{\pi}}$ is supermodular, i.e. for each $X\subseteq Y\subseteq [n]$ and $j\in [n]\setminus Y$ the inequality $F^{\pmb{\pi}}(Y\cup\{j\})-F^{\pmb{\pi}}(Y)\ge F^{\pmb{\pi}}(X\cup\{j\})-F^{\pmb{\pi}}(X)$ holds.
\end{thm}
\begin{proof}
	See the Appendix.
\end{proof}	

The greedy algorithm considers all possible $\pmb{\pi}\in \Pi$. For each fixed $\pmb{\pi}$ we start with $X=\emptyset$ and  greedily add the elements $i\in [n]\setminus X$ to $X$ as long as an improvement is possible, i.e. if there is $i\in [n]\setminus X$ such that $F^{\pmb{\pi}}(X\cup\{i\})<F^{\pmb{\pi}}(X)$. 
For a fixed choice of $\pmb{\pi}$, the greedy algorithm thus evaluates the objective $F^{\pmb{\pi}}$ $O(np)$ many times. In total, there are therefore $O(n^3p)$ calls to $F^{\pmb{\pi}}$.

Theorem~\ref{thm-sup} allows us to reduce the search space of the algorithm. Namely if adding $j$ to $X$ does not decrease the value of $F^{\pmb{\pi}}$ at some step, then adding $j$ to the current solution in the subsequent steps also cannot improve the current solution. Hence $j$ can be removed from the further considerations. 

 \begin{algorithm}[htb]
\begin{small}
$val^* \leftarrow \infty$, $X^* \leftarrow \emptyset$\;
\ForEach{$\{(\hat{c}_k,\hat{c}_l)\in \mathcal{A} \times \mathcal{B} : \hat{c}_k \le \hat{c}_l\}$}{
\lForEach{$\alpha\in \mathcal{C}$}{$\pi(\alpha) \leftarrow \max\{\hat{c}_k,\min\{\alpha,\hat{c}_l\}\}$}
$X \leftarrow \emptyset$, $bestval \leftarrow F^{\pmb{\pi}}(X)$, $bestX \leftarrow X$, $improve \leftarrow true$\;
\While{$improve = true$ and  $|X|<p$}{
$improve \leftarrow false$\;
\ForEach{$i\in[n]\setminus X$}
{
$Y \leftarrow X\cup\{i\}$\;
\If{$F^{\pmb{\pi}}(Y) \le bestval$}{
$bestX \leftarrow Y$\;
$bestval \leftarrow F^{\pmb{\pi}}(Y)$\;
$improve \leftarrow true$\;
}}
$X \leftarrow bestX$\;
}
\If{$bestval \le val^*$}{
$val^* \leftarrow bestval$\;
$X^* \leftarrow bestX$\;
}}
\Return{$X^*$}
  \caption{A greedy algorithm for \textsc{RTsR Selection}.}
 \label{alg2sssel}
\end{small} 
\end{algorithm}

In order to illustrate the algorithm consider an instance  of the problem shown in Table~\ref{tab:ex2}.
\begin{table}[htb]
\begin{center}
\caption{Example problem for the greedy algorithm with $n = 4$ and $p = 3$.}\label{tab:ex2}
\begin{tabular}{c|ccc}
$i$ & $C_i$ & $\underline{c}_i$ & $\overline{c}_i$ \\
\hline
1 & 6 & 9 & 13 \\
2 & 1 & 1 & 4 \\
3 & 4 & 2 & 12 \\
4 & 12 & 2 & 6
\end{tabular}
\end{center}
\end{table}
We get $\mathcal{A}=\{1,2,4,6,9,12\}$, $\mathcal{B}=\{1,2,4,6,9,12,13\}$ and $\mathcal{C}=\{1,2,4,6,9,12,13\}$. For $\hat{c}_k=2$ and $\hat{c}_l = 6$, we get $\pi(1)=2$, $\pi(2)=2$, $\pi(4)=4$, $\pi(6)=6$, $\pi(9)=6$, $\pi(12)=6$, and $\pi(13)=6$. For this vector $\pmb{\pi}$ we compute ${\rm P}(\pmb{\pi})$ by solving the following problem:
\begin{align*}
{\rm P}(\pmb{\pi})=\min\ &z \\
\text{s.t. } & z \ge -2 + 5x_1 + 0x_2 + 3x_3 + 11x_4  \\
& z \ge \phantom{-}1 + 4x_1 -1x_2 + 2x_3 + 10x_4 \\
& z \ge \phantom{-}3 + 2x_1 -3x_2 + 2x_3 + 10x_4 \\
& z \ge \phantom{-}5 + 0x_1 -3x_2 + 0x_3 + 10x_4 \\
& z \ge \phantom{-}8 -3x_1 -3x_2 -3x_3 + 10x_4 \\
& z \ge 11 -6x_1 -3x_2 -6x_3 + 10x_4 \\
& z \ge 11 -7x_1 -3x_2 -6x_3 + 10x_4 \\
& x_1 + x_2 + x_3 + x_4 \le 3 \\
& x_1,x_2,x_3,x_4 \in \{0,1\}
\end{align*}
An optimal solution to this problem is $\pmb{x}=(0,1,1,0)$ with ${\rm P}(\pmb{\pi})=2$. In fact, this is an optimal first-stage solution to the sample instance with $Z(\pmb{x})=2$. Figure~\ref{fig:exgreedy} shows the search space of the greedy algorithm for the fixed $\pmb{\pi}$. If  1  is chosen in the first step, then the best achievable regret value is 4 by using the first-stage solution $X=\{1,2\}$. However, the optimal regret value is 2 by using the first-stage solution $X=\{2,3\}$.
So, the example demonstrates that the approximation ratio of the greedy algorithm is at least~2. We conjecture that  the algorithm is indeed a 2-approximation, so the example presented is a worst one. The search space can be reduced by applying Theorem~\ref{thm-sup}. For example, 4 need only be considered in the first step, because $F^{\pmb{\pi}}(\emptyset)<F^{\pmb{\pi}}(\{4\})$. So adding 4 in the next steps cannot decrease the value of the current solution and~4 can be removed from further considerations. Also we need not to consider adding~3 to $\{1,2\}$, because by Theorem~\ref{thm-sup}, $F^{\pmb{\pi}}(\{1,2\}\cup \{3\})-F^{\pmb{\pi}}(\{1,2\})\geq F^{\pmb{\pi}}(\{1\}\cup\{3\})-F^{\pmb{\pi}}(\{1\})>0$.

\begin{figure}[htb]
\begin{center}
\includegraphics[width=0.4\textwidth]{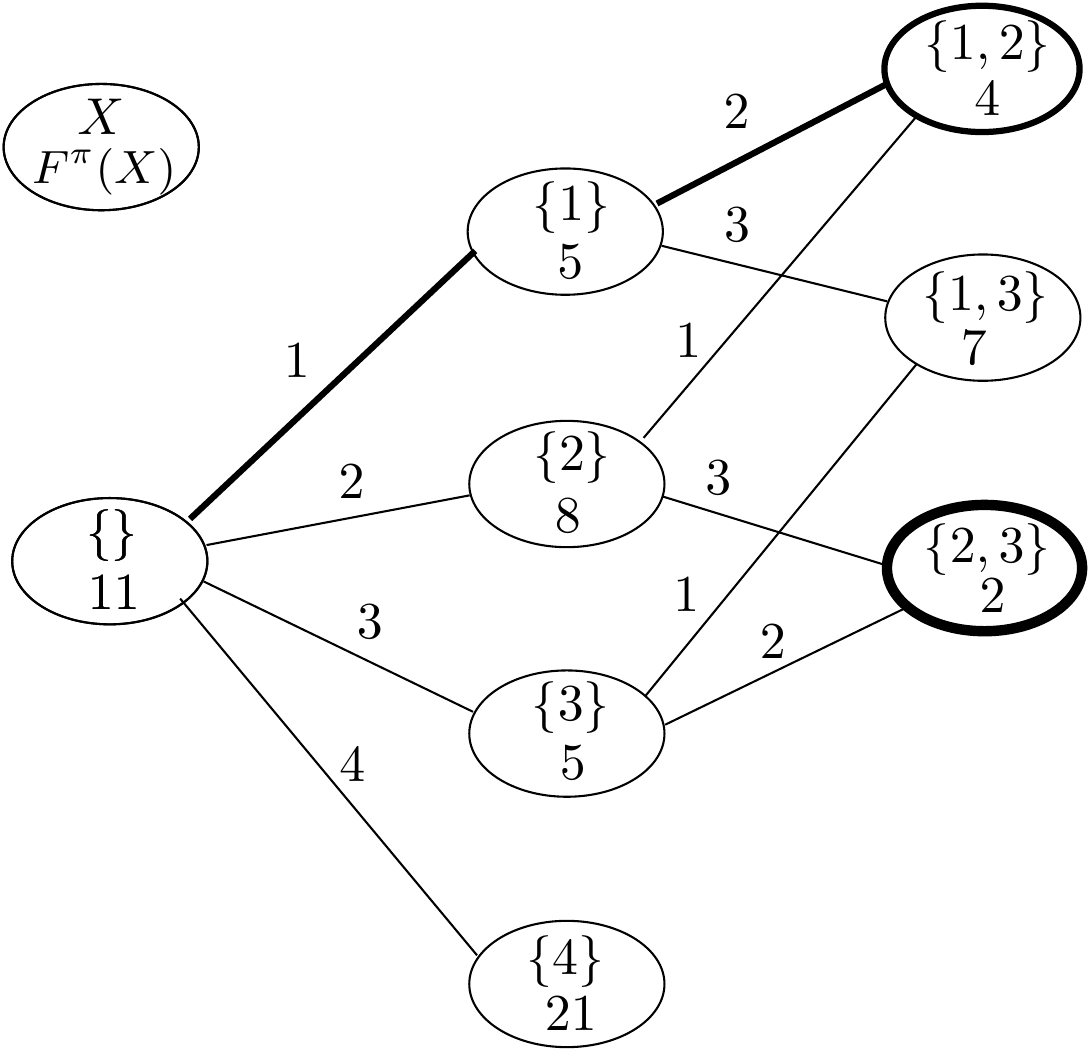}
\caption{Greedy search space for the sample problem.}\label{fig:exgreedy}
\end{center}
\end{figure}

The example suggest an improvement of the greedy algorithm. Observe, that we can achieve the optimal solution by adding 3 instead of 1 to $X=\emptyset$. So it may be advantageous to start from all possible subsets of $[n]$ with up to $L$ items. For a small constant $L$, the algorithm remains polynomial.

\section{Conclusions}
\label{sec:conclusions}
In this paper we have discussed a class of two-stage combinatorial optimization problems under interval uncertainty representation. We have used the maximum regret criterion to choose the best first-stage solution. The problem has different properties than the corresponding minmax regret single-stage counterpart. In particular, there is no easy characterization of a worst-case scenario for a given first-stage solution and computing its maximum regret can be NP-hard even if the deterministic problem is polynomially solvable. We have proposed a general procedure for solving the problem, which is based on a standard row and column generation technique. This method can be used to compute optimal solutions for the problems with reasonable size.
Furthermore,
we have provided a characterization of the problem complexity for two variants of  \textsc{TStR shortest Path} and proposed compact MIP formulations for \textsc{TStR Selection}. There is a number of open problems concerning the considered approach. The computational complexity of \textsc{TStR Selection} is open. It is also interesting to explore the complexity of more general class of matroidal problems, for which \textsc{Selection} is a special case (another important special case is the \textsc{Minimum Spanning Tree} problem). Finally, no approximation algorithm is known for \textsc{TStR}~$\mathcal{P}$. We have showed that the mid-point heuristic, used in the single-stage problems, does not guarantee any approximation ratio. We conjecture that the greedy algorithm proposed may be indeed a 2-approximation one for  \textsc{TStR Selection}. Proving this (or showing a counterexample) is an interesting open problem.

\subsubsection*{Acknowledgements}
The second and third author were supported by
 the National Science Centre, Poland, grant 2017/25/B/ST6/00486.

%


\section*{Appendix}

To prove Theorem~\ref{thm-sup}, we first derive two lemmas.

\begin{lem}\label{w-lemma}
The coefficients $\omega^{\pmb{\pi}}_i(\alpha)$, $i\in [n]$, are nonincreasing in $\alpha$.
\end{lem}\begin{proof}
Recall that
\begin{align*}
\omega^{\pmb{\pi}}_i(\alpha) &= C_i - \alpha + [\alpha-\overline{c}_i]_+ + \overline{\rho}^{\pmb{\pi}}_i(\alpha) - \underline{\rho}^{\pmb{\pi}}_i(\alpha) 
\end{align*}
where
\begin{align*}
\underline{\rho}^{\pmb{\pi}}_i(\alpha) &= \max\{0,\pi(\alpha) - C_i, \pi(\alpha) - \underline{c}_i - [\alpha-\underline{c}_i]_+ + [\alpha-\overline{c}_i]_+ \} \\
\overline{\rho}^{\pmb{\pi}}_i(\alpha) &= \max\{ 0 , \pi(
\alpha) - C_i, \pi(\alpha) - \underline{c}_i\}
\end{align*}
and 
$$\pi(\alpha)=\max\{\hat{c}_k,\min\{\alpha,\hat{c}_l\}\}$$
 for some fixed $\hat{c}_k\in \mathcal{A}$ and $\hat{c}_l\in \mathcal{B}$, $\hat{c}_k\leq \hat{c}_l$.
 We consider the following cases:
\begin{enumerate}
\item $\alpha \le \hat{c}_k$:  In this case $\pi(\alpha)=\hat{c}_k$ and $\overline{\rho}^{\pmb{\pi}}_i(\alpha)=\overline{\rho}^{\pmb{\pi}}_i$ is constant. Hence

\begin{align*}
\omega^{\pmb{\pi}}_i(\alpha) &= 
 \begin{cases}
C_i -\alpha + \overline{\rho}^{\pmb{\pi}}_i - \max\{0,\hat{c}_k - C_i, \hat{c}_k - \underline{c}_i\} & \text{ if } \alpha\le \underline{c}_i \\
C_i -\alpha+ \overline{\rho}^{\pmb{\pi}}_i -  \max\{0,\hat{c}_k - C_i, \hat{c}_k - \alpha\} & \text{ if } \underline{c}_i < \alpha< \overline{c}_i \\
C_i - \overline{c}_i+\overline{\rho}^{\pmb{\pi}}_i -\max\{0,\hat{c}_k - C_i, \hat{c}_k - \overline{c}_i\} & \text{ if } \overline{c}_i \le \alpha
\end{cases}
\end{align*}
and $\omega^{\pi}_i(\alpha)$ is nonincreasing in $[0,\hat{c}_k]$

\item $\alpha\geq \hat{c}_l$: This case is analogous to the case 1 with $\pi(\alpha)=\hat{c}_l$ and constant $\overline{\rho}^{\pmb{\pi}}(\alpha)$, so $\omega^{\pmb{\pi}}_i(\alpha)$  is nonincreasing in $[\hat{c}_l,\infty)$

\item $\hat{c}_k \le \alpha \le \hat{c}_l$: Then $\pi(\alpha)=\alpha$ and we distinguish further:
\begin{enumerate}
\item $C_i \le \underline{c}_i \le \overline{c}_i$: Then
\[ \overline{\rho}^{\pmb{\pi}}_i(\alpha) - \underline{\rho}^{\pmb{\pi}}_i(\alpha) = [\alpha-C_i]_+ - [\alpha-C_i]_+ = 0 \]
and $\omega^{\pmb{\pi}}_i(\alpha)=C_i - \alpha + [\alpha-\overline{c}_i]_+$ is nonincreasing in $[\hat{c}_k,\hat{c}_l]$.
\item $\underline{c}_i \le C_i \le \overline{c}_i$: Then
\[ \overline{\rho}^{\pmb{\pi}}_i(\alpha) - \underline{\rho}^{\pmb{\pi}}_i(\alpha) = [\alpha-\underline{c}_i]_+ - [\alpha-C_i]_+\]
and therefore
\begin{align*}
\omega^{\pmb{\pi}}_i(\alpha) &= C_i - \alpha + [\alpha-\overline{c}_i]_+ + [\alpha-\underline{c}_i]_+ - [\alpha-C_i]_+ \\
&= \begin{cases}
C_i -\alpha & \text{ if } \alpha \le \underline{c}_i \\
C_i - \underline{c}_i & \text{ if } \underline{c}_i < \alpha \le C_i \\
C_i -\underline{c}_i - \alpha + C_i & \text{ if } C_i < \alpha \le \overline{c}_i  \\
C_i -\overline{c}_i - \underline{c}_i + C_i & \text{ if } \overline{c}_i <\alpha \\
\end{cases}
\end{align*}
and
$\omega^{\pmb{\pi}}_i(\alpha)$ is nonincreasing in $[\hat{c}_k,\hat{c}_l]$.

\item $\underline{c}_i \le \overline{c}_i \le C_i$: Then
\begin{align*}
\omega^{\pmb{\pi}}_i(\alpha) &= C_i - \alpha + [\alpha-\overline{c}_i]_+ + [\alpha-\underline{c}_i]_+ - [\alpha-\underline{c}_i - [\alpha-\underline{c}_i]_+ + [\alpha-\overline{c}_i]_+ ]_+ \\
&= \begin{cases}
C_i - \alpha & \text{ if } \alpha \le \underline{c}_i \\
C_i - \underline{c}_i & \text{ if } \underline{c}_i < \alpha <\overline{c}_i \\
C_i + \underline{c}_i & \text{ if } \overline{c}_i \le \alpha
\end{cases}
\end{align*}
\end{enumerate}
\end{enumerate}
and
$\omega^{\pmb{\pi}}_i(\alpha)$ is nonincreasing in $[\hat{c}_k,\hat{c}_l]$.
\end{proof}

\begin{lem}\label{u-lemma}
For any functions $f,g, h: [a,b] \to \mathbb{R}$,  attaining a maximum in $[a,b]$, where $g$ and $h$ are nonincreasing in $[a,b]$, it holds that
\[ \max_u \Big( f(u) + g(u) \Big)
+ \max_u \Big( f(u) + h(u) \Big) 
\le
\max_u f(u)
+ \max_u \Big( f(u) + g(u) + h(u) \Big). \]
\end{lem}\begin{proof}
Let $u_1$, $u_2$, $u_3$ and $u_4$ be such that
\begin{align*}
\max_u \Big( f(u) + g(u) \Big) &= f(u_1) + g(u_1) \\
\max_u \Big( f(u) + h(u) \Big) &= f(u_2) + h(u_2) \\
\max_u \Big( f(u) \Big) &= f(u_3) \\
\max_u \Big( f(u) + g(u) + h(u) \Big) &= f(u_4) + g(u_4) + h(u_4).
\end{align*}
We can assume without loss of generality that $u_2 \le u_1$ (otherwise, we exchange $h$ and $g$). Then  $g(u_2)\ge g(u_1)$ by the monotonicity of $g$, $f(u_1) \le f(u_3) $, because $u_3$ maximizes $f$ and
$f(u_1)+g(u_1)+f(u_2)+h(u_2)\leq f(u_3)+g(u_2)+f(u_2)+h(u_2)\leq f(u_3)+g(u_4)+f(u_4)+h(u_4).$
\end{proof}

We are now in the position to prove that $F^{\pmb{\pi}}$ is indeed supermodular.

\begin{proof}[Proof of Theorem~\ref{thm-sup}.]
Let any $X\subseteq Y\subseteq[n]$ and $j\in[n]\setminus Y$ be given such that $|Y|+1\le p$. We need to show that $F^{\pmb{\pi}}(X\cup\{j\}) - F^{\pmb{\pi}}(X) \le F^{\pmb{\pi}}(Y\cup\{j\}) - F^{\pmb{\pi}}(Y)$. This is the case if and only if
\begin{align*}
& \max_{\alpha} \left( \nu^{\pmb{\pi}}(\alpha) + \sum_{i\in X} \omega^{\pmb{\pi}}_i(\alpha) + \omega^{\pmb{\pi}}_j(\alpha) \right) 
-  \max_{\alpha} \left( \nu^{\pmb{\pi}}(\alpha) + \sum_{i\in X} \omega^{\pmb{\pi}}_i(\alpha) \right) \\
\le &  \max_{\alpha} \left( \nu^{\pmb{\pi}}(\alpha) + \sum_{i\in Y} \omega^{\pmb{\pi}}_i(\alpha) + \omega^{\pi}_j(\alpha) \right) 
-  \max_{\alpha} \left( \nu^{\pmb{\pi}}(\alpha) + \sum_{i\in Y} \omega^{\pmb{\pmb{\pi}}}_i(\alpha) \right). 
\end{align*}
Fix $f(\alpha)=\nu^{\pmb{\pi}}(\alpha) + \sum_{i\in X} \omega^{\pmb{\pi}}_i(\alpha)$, $g(\alpha)=\omega^{\pmb{\pi}}_j(\alpha)$ and $h(\alpha)=\sum_{i\in Y\setminus X} \omega^{\pmb{\pi}}_i(\alpha)$. The inequality can be rewritten equivalently as 
$$\max_{\alpha}(f(\alpha)+g(\alpha))-\max_{\alpha} f(\alpha)\leq \max_{\alpha}(f(\alpha)+h(\alpha)+g(\alpha))-\max_{\alpha}(f(\alpha)+h(\alpha)).$$
According to Lemma~\ref{w-lemma}, the functions $g(\alpha)$ and $h(\alpha)$ are nonincreasing. Applying Lemma~\ref{u-lemma}, the theorem thus follows.
\end{proof}

\end{document}